\documentclass[a4paper, USenglish, cleveref, autoref, numberwithinsect, thm-restate]{lipics-v2021}
\pdfoutput=1 \hideLIPIcs

\title{The Parameterized Complexity of Learning Monadic Second-Order Logic}

\author{Steffen {van Bergerem}}{Humboldt-Universität zu Berlin, Germany}{steffen.van.bergerem@informatik.hu-berlin.de}{https://orcid.org/0000-0002-5212-8992}{This work was funded by the Deutsche Forschungsgemeinschaft (DFG, German Research Foundation) -- project number 431183758 (gefördert durch die Deutsche Forschungsgemeinschaft (DFG) -- Projektnummer 431183758).}

\author{Martin Grohe}{RWTH Aachen University, Germany}{grohe@informatik.rwth-aachen.de}{https://orcid.org/0000-0002-0292-9142}{Funded by the European Union (ERC, SymSim, 101054974). Views and opinions expressed are however those of the author(s) only and do not necessarily reflect those of the European Union
or the European Research Council. Neither the European Union nor the granting authority can be held responsible for them.}

\author{Nina Runde}{RWTH Aachen University, Germany}{runde@lics.rwth-aachen.de}{https://orcid.org/0009-0000-4547-1023}{This work was funded by the Deutsche Forschungsgemeinschaft (DFG, German Research Foundation) -- project number 453349072 (gefördert durch die Deutsche Forschungsgemeinschaft (DFG) -- Projektnummer 453349072).}

\authorrunning{S. van Bergerem, M. Grohe, and N. Runde}

\Copyright{Steffen van Bergerem, Martin Grohe, and Nina Runde}



\category{} 

\relatedversion{}
\nolinenumbers

\EventEditors{J\"{o}rg Endrullis and Sylvain Schmitz}
\EventNoEds{2}
\EventLongTitle{33rd EACSL Annual Conference on Computer Science Logic (CSL 2025)}
\EventShortTitle{CSL 2025}
\EventAcronym{CSL}
\EventYear{2025}
\EventDate{February 10--14, 2025}
\EventLocation{Amsterdam, Netherlands}
\EventLogo{}
\SeriesVolume{326}
\ArticleNo{8}

\usepackage{tikz}
\usetikzlibrary{backgrounds,calc,shadows,shapes.geometric}
\tikzset{dropshadow/.style={drop shadow={opacity=.4, shadow xshift=.25ex, shadow yshift=-.25ex}},
  vertex/.style={draw, semithick, circle, inner sep=.8ex, fill=white, dropshadow},
  small vertex/.style={draw, semithick, circle, inner sep=.3ex, fill=white, dropshadow},
  positive example/.style={circle, inner sep=1.8ex, fill=ibm-indigo},
  negative example/.style={circle, inner sep=1.8ex, fill=ibm-orange},
  edge/.style={semithick},
}

\newcommand\problemDefPara[4]{\vspace{0.3cm}
\boxed{
\begin{minipage}{0.9\textwidth}
  \noindent \textsc{{#1}} \vspace{0.2cm}\\
    {\bfseries Instance}:    #2 \\
    {\bfseries Parameter}:   #3 \\
    {\bfseries Problem}:  #4
\end{minipage}
}
\vspace{0.3cm}
}

\newcommand\problemDefParaNoName[3]{\boxed{
\begin{minipage}{0.9\columnwidth}
  {\bfseries Instance}:  #1 \\
  {\bfseries Parameter}: #2 \\
  {\bfseries Problem}:   #3
\end{minipage}
}
}
\definecolor{ibm-ultramarine}{HTML}{648fff}
\definecolor{ibm-indigo}{HTML}{785ef0}
\definecolor{ibm-magenta}{HTML}{dc267f}
\definecolor{ibm-orange}{HTML}{fe6100}
\definecolor{ibm-gold}{HTML}{ffb000}

\usepackage{xspace}
\newcommand{\problemFont}[1]{\ensuremath{\textup{\textsc{#1}}}\xspace}
\newcommand{\MSOLearn}{\problemFont{MSO-Consistent-Learn}}
\newcommand{\uMSOLearn}{\problemFont{1D-MSO-Consistent-Learn}}

\newcommand{\PacMSOLearn}{\problemFont{MSO-PAC-Learn}}
\newcommand{\MSOMc}{\problemFont{MSO-Mc}}
\newcommand{\FOMc}{\problemFont{FO-Mc}}

\newcommand{\ComplexityClassFont}[1]{\ensuremath{\textup{\textsf{#1}}}\xspace}

\newcommand{\FPT}{\ensuremath{\ComplexityClassFont{FPT}}}
\newcommand{\AWstar}{\ensuremath{\ComplexityClassFont{AW}[*]}}
\newcommand{\Wone}{\ensuremath{\ComplexityClassFont{W}[1]}}

\newcommand{\LogicFont}[1]{\ComplexityClassFont{#1}}
\newcommand{\FO}{\LogicFont{FO}}
\newcommand{\MSO}{\LogicFont{MSO}}
\newcommand{\Lor}{\bigvee}
\newcommand{\Land}{\bigwedge}
\newcommand{\type}[3]{\ensuremath{\textup{tp}^{#2}_{#1}(#3)}}

\renewcommand{\phi}{\varphi}
\renewcommand{\epsilon}{\varepsilon}

\usepackage{BOONDOX-cal}

\newcommand{\Structure}[1]{\ensuremath{\mathcal{#1}}}

\newcommand{\T}{\Structure{T}}

\newcommand{\CC}{{\mathcal C}}

\newcommand{\CR}{{\mathcal R}}

\newcommand{\Ca}{\mathcal a}

\usepackage{mathtools}
\newcommand{\deff}{\coloneqq}
\newcommand{\abs}[1]{\left\lvert#1\right\rvert}

\DeclareMathOperator{\cw}{\ensuremath{\textsf{cw}}}
\DeclareMathOperator*{\qr}{qr}

\DeclareMathOperator*{\VC}{VC}

\newcommand{\Tp}{\mathsf{Tp}}
\newcommand{\tp}{\mathsf{tp}}
\newcommand{\btheta}{{\vec \theta}}

\renewcommand{\vec}[1]{\bar{#1}}

\newcommand{\set}[1]{\ensuremath{\{#1\}}}
\newcommand{\setc}[2]{\ensuremath{\set{#1 \mid #2}}}
\newcommand{\bigset}[1]{\ensuremath{\bigl\{ #1 \bigr\}}}
\newcommand{\bigsetc}[2]{\bigset{#1 \bigmid #2}}

\newcommand{\bigmid}{\mathrel{\big|}}

\newcommand{\NN}{\mathbb{N}}
\newcommand{\Nat}{{\mathbb N}}

\newcommand{\bigO}{\mathcal{O}}

\newcommand{\X}{\mathbb{X}}
\newcommand{\TS}{S}
\newcommand{\D}{\mathcal{D}}
\newcommand{\Hypo}{\mathcal{H}}
\DeclareMathOperator{\err}{err}

\newcommand{\ie}{i.\,e.}

\newcommand{\eg}{e.\,g.}

\begin{document}

\maketitle

\begin{abstract}
Within the model-theoretic framework for supervised learning
introduced by Grohe and Tur{\'{a}}n (TOCS 2004),
we study the parameterized complexity of learning concepts
definable in monadic second-order logic (\(\MSO\)).
We show that the problem of learning an \(\MSO\)-definable concept
from a training sequence of labeled examples is fixed-parameter tractable
on graphs of bounded clique-width, and that it is hard for the
parameterized complexity class para-NP on general graphs.

It turns out that an important distinction to be made is between
$1$-dimensional and higher-dimensional concepts,
where the instances of a $k$-dimensional concept are
$k$-tuples of vertices of a graph.
For the higher-dimensional case, we give a learning algorithm
that is fixed-parameter tractable in the size of the graph,
but not in the size of the training sequence, and we give a
hardness result showing that this is optimal.
By comparison, in the 1-dimensional case, we obtain
an algorithm that is fixed-parameter tractable in both.
\end{abstract}

\clearpage
\section{Introduction}

We study abstract machine-learning problems in a logical framework with a declarative view on learning, where the (logical) specification of concepts is separated from the choice of specific machine-learning models and algorithms (such as neural networks). Here we are concerned with the  computational complexity of learning problems in this logical learning framework, that is, the \emph{descriptive complexity of learning} \cite{van_bergerem_thesis_2023}.

Specifically, we consider Boolean classification problems that can be specified in monadic second-order logic (\(\MSO\)).
The input elements for the classification task come from a set \(\X\),
the \emph{instance space}.
A \emph{classifier} on \(\X\) is a function
\(c \colon \X \to \set{+,-}\).
Given a \emph{training sequence} \(\TS\) of labeled examples
\((x, \lambda) \in \X \times \set{+,-}\),
we want to find a classifier, called a \emph{hypothesis},
that explains the labels given in \(\TS\)
and that can also be used to predict the labels of elements
from \(\X\) not given as examples. In the logical setting, the instance space \(\X\) is a set of tuples from a (relational) structure,
called the \emph{background structure},
and classifiers are described by formulas of some logic, in our case \MSO, using parameters from the background structure.
This model-theoretic learning framework was introduced by Grohe and
Tur{\'{a}}n~\cite{grohe_learnability_2004} and further studied
in~\cite{grienenberger_learning_2019,grohe_learning_2017,grohe_learning_2017-1,
  van_bergerem_learning_2019,van_bergerem_parameterized_2022,
van_bergerem_learning_2021,van_bergerem_thesis_2023}.

We study these problems within the following well-known settings
from computational learning theory.
In the \emph{consistent-learning} model,
the examples are assumed to be generated using an unknown classifier,
the \emph{target concept}, from a known \emph{concept class}.
The task is to find a hypothesis that is consistent with the training sequence \(\TS\),
i.e.\ a function \(h \colon \X \to \set{+,-}\) such that
\(h(x) = \lambda\) for all \((x, \lambda) \in \TS\).
In Haussler's model of \emph{agnostic probably approximately correct (PAC) learning}~\cite{haussler_pac_1992},
a generalization of Valiant's \emph{PAC learning} model~\cite{valiant_theory_1984},
an (unknown) probability distribution \(\mathcal{D}\)
on \(\X \times \set{+,-}\) is assumed,
and training examples are drawn independently from this distribution.
The goal is to find a hypothesis that generalizes well,
i.e.\ one is interested in algorithms that
return with high probability a hypothesis with a small expected error
on new instances drawn from the same distribution.
For more background on PAC learning,
we refer to~\cite{kearns_introduction_1994,MohriRT18,shalev-shwartz_understanding_2014}.
In both settings, we require our algorithms to return
a hypothesis from a predefined \emph{hypothesis class}.

\paragraph*{Our Contributions}
In this paper, we study the parameterized complexity of the consistent-learning problem
\(\MSOLearn\) and the PAC-learning problem \(\PacMSOLearn\).
In both problems, we are given a graph \(G\) (the background structure)
and a sequence of labeled training examples of the form \((\vec{v},\lambda)\),
where \(\vec{v}\) is a \(k\)-tuple of vertices from \(G\) and \(\lambda \in \set{+, -}\).
The goal is to find a hypothesis of the form \(h_{\phi, \bar{w}}\)
for an \(\MSO\) formula \(\phi(\bar{x}; \bar{y})\) and a tuple \(\bar{w}\)
with \(h_{\phi, \bar{w}}(\bar{v}) \deff +\) if \(G \models \phi(\bar{v}; \bar{w})\)
and \(h_{\phi, \bar{w}}(\bar{v}) \deff -\) otherwise.
For \(\MSOLearn\), this hypothesis should be consistent with the given training examples.
For \(\PacMSOLearn\), the hypothesis should generalize well.
We restrict the complexity of allowed hypotheses by giving a bound \(q\) on the quantifier rank of \(\phi\)
and a bound \(\ell\) on the length of \(\bar{w}\).
Both \(q\) and \(\ell\) as well as the dimension \(k\) of the problem,
that is, the length of the tuples to classify,
are part of the parameterization of the problems.
A detailed description of \(\MSOLearn\) is given in \cref{sec:tract}.
The problem \(\PacMSOLearn\) is formally introduced in \cref{sec:pac}.

\begin{example}
\label{example_bipartite}
Assume we are given the graph $G$ depicted in \cref{fig_example},
the training sequence
${S = ((v_1,+),(v_3,+),(v_4,-),(v_5,-))}$, and
${k=1}$, ${\ell=1}$, $q=3$. Note that $k=1$ indicates that the instances are vertices of the input graph $G$. Furthermore, $\ell=1$ indicates that the specification may involve one vertex of the input graph as a parameter.
Finally, $q=3$ indicates that the formula specifying the hypothesis must have quantifier rank at most $3$.

Our choice of a hypothesis $h \colon V(G)\to\{+,-\}$ consistent with $S$ says that “there is a bipartite partition of the graph
such that all positive instances $x$ are on the same side as $v_2$ and all negative examples are on the other side.”
This hypothesis can be formally specified in $\MSO$ as  $h_{\phi,\bar{w}}$
for the \(\MSO\) formula
\(\phi(x;y) = \exists Z \bigl(\psi_{\text{bipartite}}(Z) \land Z(x) \land Z(y)\bigr)\)
and parameter setting \(\vec{w} = (v_2)\),
where \(\psi_{\text{bipartite}}(Z) = \forall z_1 \forall z_2
\Bigl(E(z_1,z_2) \rightarrow \neg \bigl(Z(z_1) {\leftrightarrow} Z(z_2)\bigr)\Bigr)\).
\end{example}

\begin{figure}
\centering
\begin{tikzpicture}
  \node[small vertex] (1) at (0, 1.15) {$v_1$};
\node[small vertex] (2) at (2, 1.15) {$v_2$};
\node[small vertex] (3) at (4, 1.15) {$v_3$};
\node[small vertex] (4) at (0, 0) {$v_4$};
\node[small vertex] (5) at (2, 0) {$v_5$};
\node[small vertex] (6) at (4, 0) {$v_6$};

\begin{pgfonlayer}{background}
  \node[below of=1, node distance=0, positive example] {};
  \node[below of=3, node distance=0, positive example] {};
  \node[below of=4, node distance=0, negative example] {};
  \node[below of=5, node distance=0, negative example] {};
\end{pgfonlayer}

\draw[edge] (3) to (6);
\draw[edge] (1) to (4);
\draw[edge] (1) to (5);
\draw[edge] (1) to (6);
\draw[edge] (2) to (5);
\draw[edge] (2) to (6);
 \end{tikzpicture}
\caption{Graph $G$ for \cref{example_bipartite}.
Positive examples are shown in \textcolor{ibm-indigo}{purple}, and negative examples are shown in \textcolor{ibm-orange}{orange}.}
\label{fig_example}
\end{figure}

For the \(1\)-dimensional case of \(\MSOLearn\), called \(\problemFont{1D-MSO-Consistent-}\) \(\problemFont{Learn}\),
\cite{grohe_learning_2017-1,grienenberger_learning_2019} gave algorithms that are sublinear
in the background structures after a linear-time pre-pro\-cessing stage
for the case that the background structure is a string or a tree.
This directly implies that \(\uMSOLearn\) can be solved in time $f(\ell,q) \cdot n$ for some function $f$, that is, in fixed-parameter linear time, if the background structure is a string or a tree.
Here $n$ is the size of the background structure and $\ell,q$ are the parameters of the learning problem described above.
We generalize the results to labeled graphs of bounded clique-width.
Graphs of clique-width $c$ can be described by a \emph{$c$-expression}, that is, an expression in a certain graph grammar that only uses $c$ labels (see \cref{sub:prelim:clique-width} for details). In our algorithmic results for graphs of bounded clique-width, we always assume that the graphs are given in the form of a $c$-expression. We treat $c$ as just another parameter of our algorithms. By
the results of Oum and Seymour \cite{oum_approximating_2006}, we can always compute
a \(2^{\bigO(c)}\)-expression for a graph of clique-width $c$ by a fixed-parameter tractable algorithm.

\begin{restatable}{theorem}{oneDConsistentLearning}
\label{thm:1d-tract-main}
    Let \(\mathcal{C}\) be a class of labeled graphs of bounded clique-width. Then
    \(\problemFont{1D-MSO-}\) \(\problemFont{Consistent-Learn}\)
    is fixed-parameter linear on \(\mathcal{C}\).
\end{restatable}

Since graphs of bounded tree-width also have bounded clique-width,
our result directly implies fixed-parameter linearity on graph classes of bounded tree-width.

Our proof for \cref{thm:1d-tract-main} relies on the model-checking techniques due to Courcelle, Makowsky,
and Rotics for graph classes of bounded clique-width \cite{courcelle_linear_1999}.
To make use of them, we encode the training examples into the graph as new labels.
While this construction works for \(k=1\),
it fails for higher dimensions if there are too many examples to encode.

As far as we are aware, all previous results for learning $\MSO$ formulas are restricted to the one-dimensional case of the problem. We give the first results for $k>1$, presenting two different approaches that yield tractability results in higher dimensions.

As we discuss in \cref{sec:pac},
for the PAC-learning problem \(\PacMSOLearn\) in higher dimensions,
we can restrict the number of examples to consider to a constant.
In this way, we obtain fixed-parameter tractability results for learning \(\MSO\)-definable concepts in higher dimensions,
similar to results for first-order logic on nowhere dense classes~\cite{van_bergerem_parameterized_2022,van_bergerem_thesis_2023}.

\begin{restatable}{theorem}{PAC}
\label{Theorem:tractability-pac}
Let \(\mathcal{C}\) be a class of labeled graphs of bounded clique-width.
Then
\(\problemFont{MSO-PAC-}\) \(\problemFont{Learn}\)
is fixed-parameter linear on \(\mathcal{C}\).
\end{restatable}

In the second approach to higher-dimensional tractability, and as the main result of this paper,
we show in \cref{sec:high-dim-cons} that a consistent hypothesis can be learned on graphs of bounded clique-width
with a quadratic running time in terms of the size of the graph.

\begin{restatable}{theorem}{highDimTract}
\label{thm:tract-high-dim}
  There is a function $g \colon \NN^5 \to \NN$ such that,
  for a \(\Lambda\)-labeled graph $G$ of clique-width $\cw(G) \leq c$
  and a training sequence \(S\) of size $\abs S = m$,
  the problem
  \(\problemFont{MSO-Consistent-}\) \(\problemFont{Learn}\)
  can be solved in time
  \[\bigO \bigl((m+1)^{g(c, \abs{\Lambda}, q, k, \ell)}|V(G)|^2\bigr).\]
\end{restatable}

While this is not strictly a fixed-parameter tractability result,
since we usually do not consider $m$ to be part of the parameterization,
we show in \cref{sec:high-dim-hardness} that this bound is optimal. Technically, this result is much more challenging than Theorems~\ref{Theorem:tractability-pac} and \ref{thm:1d-tract-main}. While we still use an overall dynamic-programming strategy that involves computing \(\MSO\) types, here we need to consider MSO types over sequences of tuples. The number of such sequence types is not constantly bounded, but exponential in the length of the sequence. The core of our argument is to prove that the number of relevant types can be polynomially bounded. This fundamentally distinguishes our approach from typical MSO/automata arguments, where types are from a bounded set (and they correspond to the states of a finite automaton).

Lastly, we study $\MSOLearn$ on arbitrary classes of labeled graphs.
Analogously to the hardness of learning \(\FO\)-definable concepts and the relation to the \(\FO\)-model-checking problem discussed in \cite{van_bergerem_parameterized_2022},
we are interested
specifically in the relation of \MSOLearn to the \MSO-model-checking problem \MSOMc.
We show that \MSOMc can already be reduced to the 1-dimensional case of \MSOLearn, even with a training sequence of size two.
This yields the following hardness result that we prove in \cref{sec:hardness}.

\begin{theorem}
\label{theorem:reduction}
\uMSOLearn is para-NP-hard under fpt Turing reductions.
\end{theorem}

\paragraph*{Related Work}
The model-theoretic learning framework studied in this paper was introduced in \cite{grohe_learnability_2004}.
There, the authors give information-theoretic learnability results for hypothesis classes
that can be defined using first-order and monadic second-order logic on restricted classes of
structures.

Algorithmic aspects of the framework were first studied in~\cite{grohe_learning_2017},
where it was proved that concepts definable in first-order logic can be learned
in time polynomial in the degree of the background structure
and the number of labeled examples the algorithm receives as input,
independently of the size of the background structure.
This was generalized to first-order logic with counting~\cite{van_bergerem_learning_2019}
and with weight aggregation~\cite{van_bergerem_learning_2021}.
On structures of polylogarithmic degree, the results yield learning algorithms
running in time sublinear in the size of the background structure.
It was shown in~\cite{grohe_learning_2017-1,van_bergerem_learning_2019}
that sublinear-time learning is no longer possible if the degree is unrestricted.
To address this issue, in~\cite{grohe_learning_2017-1},
it was proposed to introduce a preprocessing phase where,
before seeing any labeled examples,
the background structure is converted to a data structure that supports sublinear-time learning later.
This model was applied to monadic second-order logic on strings~\cite{grohe_learning_2017-1}
and trees~\cite{grienenberger_learning_2019}.

The parameterized complexity of learning first-order logic was first studied in~\cite{van_bergerem_parameterized_2022}.
Via a reduction from the model-checking problem,
the authors show that on arbitrary relational structures,
learning hypotheses definable in \(\FO\) is \(\AWstar\)-hard.
In contrast to this, they show that the problem is fixed-parameter tractable on nowhere dense graph classes.
This result has been extended to nowhere dense structure classes in~\cite{van_bergerem_thesis_2023}.
Although not stated as fpt results, the results in \cite{grohe_learning_2017-1,grienenberger_learning_2019}
yield fixed-parameter tractability for learning \(\MSO\)-definable concepts on strings and trees
if the problem is restricted to the 1-dimensional case where the tuples to classify are single vertices.

The logical learning framework is related to, but different from the framework of \emph{inductive logic programming}
(see, \eg, \cite{cohen_polynomial_1995,muggleton_inductive_1991,muggleton_inductive_1994}),
which may be viewed as the classical logic-learning framework.
In the database literature, there are various approaches to learning queries from examples
\cite{barcelo_conjunctive_2017,barcelo_conjunctive_2021,haussler_conjuctive_1989,hirata_conjunctive_2000,kimelfeld_conjunctive_2018,bonifati_path_queries_2015,staworko_twig_2012,abouzied_quantified_boolean_2013,aizenstein_hardness_1998,bonifati_join_2016,sloan_boolean_2010,ten_cate_conjunctive_2021}.
Many of these are concerned with active learning scenarios, whereas we are in a statistical learning setting.
Moreover, most of the results are concerned with conjunctive queries or queries outside the relational database model,
whereas we focus on monadic second-order logic.
Another related subject in the database literature is the problem of learning schema mappings from examples
\cite{alexe_schema_mappings_2011,bonifati_interactive_2019,ten_cate_schema_mappings_2013,ten_cate_active_2018,gottlob_schema_mapping_2010}.
In formal verification, related logical learning frameworks~\cite{champion_ice_2020,ezudheen_ice_2018,garg_ice_2014,loeding_abstract_2016,zhu_chc_2018}
have been studied as well.
In algorithmic learning theory, related works study the parameterized complexity of
several learning problems~\cite{arvind_kjuntas_2007,li_learning_2018} including,
quite recently, learning propositional CNF and DNF formulas
and learning solutions to graph problems in the PAC setting~\cite{brand_parameterized_pac_2023}.
 \section{Preliminaries}
\label{sec:prelim}

We let \(\NN\) denote the set of non-negative integers.
For \(m, n \in \NN\), we let \([m,n] \deff \setc{\ell \in \NN}{m \leq \ell \leq n}\)
and \([n] \deff [1,n]\).
For a set \(V\), we let \(2^V \deff \setc{V'}{V' \subseteq V}\).

\subsection{Clique-Width}
\label{sub:prelim:clique-width}
In this paper, graphs are always undirected and simple (no loops or parallel edges);
we view them as \(\set{E}\)-structures for a binary relation symbol \(E\),
and we denote the set of vertices of a graph \(G\) by \(V(G)\).
A \emph{label set} is a set \(\Lambda\) of unary relation symbols,
and a \emph{\(\Lambda\)-graph} or \emph{\(\Lambda\)-labeled graph}
is the expansion of a graph to the vocabulary \(\set{E} \cup \Lambda\).
A \emph{labeled graph} is a \(\Lambda\)-graph for any label set \(\Lambda\).

In the following, we define \emph{expressions} to represent labeled graphs.
A \emph{base graph} is a labeled graph of order \(1\).
For every base graph \(G\), we introduce a \emph{base expression} \(\beta\)
that \emph{represents} \(G\).
Moreover, we have the following operations.
\begin{description}
  \item[Disjoint union:]
    For disjoint \(\Lambda\)-graphs \(G_1, G_2\),
    we define \(G_1 \uplus G_2\) to be the union of \(G_1\) and \(G_2\).
    If \(G_1\) and \(G_2\) are not disjoint, then \(G_1 \uplus G_2\) is undefined.
  \item[Adding edges:]
    For a \(\Lambda\)-graph \(G\) and unary relation symbols \(P, Q \in \Lambda\) with \(P \neq Q\),
    we let \(\eta_{P,Q}(G)\) be the \(\Lambda\)-graph obtained from \(G\)
    by adding an edge between every pair of distinct vertices
    \(v \in P(G)\), \(w \in Q(G)\).
    That is,
    \(E\bigl(\eta_{P,Q}(G)\bigr) \deff E(G) \cup
      \bigsetc{(v,w),(w,v)}{v \in P(G), w \in Q(G), v\neq w}\).
  \item[Relabeling:]
    For a \(\Lambda\)-graph \(G\) and unary relation symbols \(P, Q \in \Lambda\) with \(P \neq Q\),
    we let \(\rho_{P,Q}(G)\) be the \(\Lambda\)-graph obtained from \(G\)
    by relabeling all vertices in \(P\) by \(Q\), that is,
    \(V(\rho_{P,Q}(G)) \deff V(G)\),
    \(P(\rho_{P,Q}(G)) \deff \emptyset\),
    \(Q(\rho_{P,Q}(G)) \deff Q(G) \cup P(G)\),
    and \(R(\rho_{P,Q}(G)) \deff R(G)\) for all \(R \in \Lambda \setminus\set{P,Q}\).
  \item[Deleting labels:]
    For a \(\Lambda\)-graph \(G\) and a unary relation symbol \(P \in \Lambda\),
    we let \(\delta_{P}(G)\) be the restriction of \(G\) to \(\Lambda \setminus \set{P}\),
    that is, the \((\Lambda \setminus \set{P})\)-graph obtained from \(G\)
    by removing the relation \(P(G)\).
\end{description}
We also introduce a modification of the disjoint-union operator,
namely the \emph{ordered-disjoint-union operator} \(\uplus^<\), which is used in \cref{sec:high-dim-cons} to simplify notations.
\begin{description}
    \item[Ordered disjoint union:]
To introduce this operator, we need two distinguished unary relation symbols \(P_1^<\) and \(P_2^<\).
For disjoint \(\Lambda\)-graphs \(G_1, G_2\), where we assume \(P_1^<, P_2^< \not \in \Lambda\),
we let \(G_1 \uplus^< G_2\) be the \((\Lambda \cup \set{P_1^<, P_2^<})\)-expansion
of the disjoint union \(G_1 \uplus G_2\) with \(P_1^<(G_1 \uplus^< G_2) \deff V(G_1)\)
and \(P_2^<(G_1 \uplus^< G_2) \deff V(G_2)\).
By deleting the relations \(P_1^<, P_2^<\) immediately
after introducing them in an ordered disjoint union,
we can simulate a standard disjoint-union by an ordered disjoint union, that is,
\(G_1 \uplus G_2 = \delta_{P_1^<}(\delta_{P_2^<}(G_1\uplus^< G_2))\).
\end{description}
A \emph{\(\Lambda\)-expression} is a term formed from base expressions \(\beta\),
whose label set is a subset of \(\Lambda\),
using unary operators \(\eta_{P,Q}\), \(\rho_{P,Q}\), \(\delta_P\)
for \(P, Q \in \Lambda\) with \(P \neq Q\), and the binary operator \(\uplus\).
We require \(\Lambda\)-expressions to be well-formed, that is,
all base expressions represent base graphs with mutually distinct
vertices, and the label sets fit the operators.

Every \(\Lambda\)-expression \(\Xi\) \emph{describes} a \(\Lambda'\)-graph \(G_{\Xi}\)
for some \(\Lambda' \subseteq \Lambda\).
Note that there is a one-to-one correspondence between
the base expressions in \(\Xi\) and the vertices of \(G_\Xi\).
Actually, we may simply identify the vertices of \(G_\Xi\) with the base expressions in \(\Xi\).
We let \(V_\Xi \deff V(G_{\Xi})\) be the set of these base expressions.
We may then view an expression \(\Xi\) as a tree where \(V_{\Xi}\) is the set of leaves of this tree.
We let \(\abs{\Xi}\) be the number of nodes of the tree.
We have \(\abs{G_{\Xi}} = \abs{V_{\Xi}} \leq \abs{\Xi}\).
In general, we cannot bound \(\abs{\Xi}\) in terms of \(\abs{G_{\Xi}}\),
but for every \(\Lambda\)-expression \(\Xi\), we can find a \(\Lambda\)-expression \(\Xi'\)
such that \(G_{\Xi'} = G_\Xi\) and \(\abs{\Xi'} \in \bigO(\abs{\Lambda^2} \cdot \abs{G_{\Xi}})\).

Each subexpression \(\Xi'\) of \(\Xi\) describes a labeled graph \(G_{\Xi'}\)
on a subset \(V_{\Xi'} \subseteq V_\Xi\) consisting of all base expressions in \(\Xi'\).
Note that, in general, \(G_{\Xi'}\) is not a subgraph of \(G_\Xi\).

For \(c \in \NN\), a \emph{\(c\)-expression} is a \(\Lambda\)-expression
for a label set \(\Lambda\) of size \(\abs{\Lambda} = c\).
It is easy to see that every labeled graph of order \(n\) is described by an \(n\)-expression.
The \emph{clique-width} \(\cw(G)\) of a (labeled) graph \(G\) is the least \(c\)
such that \(G\) is described by a \(c\)-expression.

We remark that our notion of clique-width differs slightly
from the one given by Courcelle and Olariu \cite{courcelle_cliquewidth_2000},
since we allow vertices to have multiple labels, and we also allow the deletion of labels.
Thus, our definition is similar to the definition of \emph{multi-clique-width} \cite{furer_multicliquewidth_2017}.
However, for our algorithmic results, the definitions are equivalent,
since we have \(\cw(G) \leq \cw'(G)\) and \(\cw'(G) \in 2^{\bigO(\cw(G))}\) for every (labeled) graph \(G\),
where \(\cw'\) is the notion of clique-width from~\cite{courcelle_cliquewidth_2000}.

\begin{lemma}[\cite{oum_approximating_2006}]
  \label{theorem-find-k-expression}
  For a graph \(G\) with \(n\) vertices and clique-width \(c' \deff \cw(G)\),
  there is an algorithm that outputs a \(c\)-expression for \(G\) where \(c = 2^{3c'+2}-1\).
  The algorithm has a running time of \(\bigO(n^9 \log n)\).
\end{lemma}

\subsection{Monadic Second-Order Logic}
\label{sec-prelim-formulas}
We consider monadic second-order (\MSO) logic,
which is a fragment of second-order logic
where we only quantify over unary relations (sets).
In \MSO, we consider two kinds of free variables,
which we call set variables (uppercase \(X,Y,X_i\))
and individual variables (lowercase \(x,y,x_i\)).
The \emph{quantifier rank} \(\qr(\phi)\) of a formula \(\phi\) is the nesting depth of its quantifiers.

Let \(\tau\) be a relational vocabulary and \(q \in \NN\).
By \(\MSO(\tau, q)\), we denote the set of all \(\MSO\) formulas
of quantifier rank at most \(q\) using only relation symbols in \(\tau\),
and we let \(\MSO(\tau) \deff \bigcup_{q} \MSO(\tau, q)\).
By \(\MSO(\tau, q,  k, s)\), we denote the set of all \(\MSO(\tau, q)\) formulas
with free individual variables in \(\set{x_1, \dots, x_k}\)
and free set variables in \(\set{X_1, \dots, X_s}\).
In particular, \(\MSO(\tau, q, 0, 0)\) denotes the set of \emph{sentences}.
Moreover, it will be convenient to separate the free individual variables
into \emph{instance variables} (\(x_1, x_2, \dots\))
and \emph{parameter variables} (\(y_1, y_2, \dots\)).
For this, we let \(\MSO(\tau, q,  k, \ell, s)\) denote the set of all \(\MSO(\tau, q)\) formulas
with free instance variables in \(\set{x_1, \dots, x_k}\),
free parameter variables in \(\set{y_1, \dots, y_\ell}\),
and free set variables in \(\set{X_1, \dots, X_s}\).
Furthermore, we write \(\phi(\vec{x}, \vec{y}, \vec{X})\) to denote that
the formula \(\phi\) has its free instance variables among the entries of \(\vec{x}\),
its free parameter variables among the entries \(\vec{y}\),
and its free set variables among the entries \(\vec{X}\).

We normalize formulas such that the set of normalized formulas in
\(\MSO(\tau, q, k, \ell, s)\) is finite,
and there is an algorithm that, given an arbitrary formula in \(\MSO(\tau, q, k, \ell, s)\),
decides if the formula is normalized, and if not, computes an equivalent normalized formula.
In the following, we assume that all formulas are normalized.

In this paper, all structures we consider will be labeled graphs for some label set \(\Lambda\).
In notations such as \(\MSO(\tau, \dots)\),
it will be convenient to write \(\MSO(\Lambda, \dots)\)
if \(\tau = \set{E} \cup \Lambda\).

\subsection{Types}
\label{subsec-prelim-types}
Let \(G\) be a \(\Lambda\)-labeled graph and \(\vec{v} \in (V(G))^k\).
The \emph{\(q\)-type of \(\vec{v}\) in \(G\)} is the set \(\tp_q^G(\vec{v})\)
of all formulas \(\phi(\vec{x}) \in \MSO(\Lambda, q, k, 0)\) such that \(G \models \phi(\vec{v})\).
A \emph{\((\Lambda,q,k)\)-type} is a set \(\theta \subseteq \MSO(\Lambda, q, k, 0)\) such that,
for each \(\phi \in \MSO(\Lambda, q, k, 0)\),
either \(\phi \in \theta\) or \(\neg\phi \in \theta\).
We denote the set of all \((\Lambda, q, k)\)-types by \(\Tp(\Lambda, q, k)\).
Note that \(\tp_q^G(\vec{v}) \in \Tp(\Lambda, q, k)\).
For a type \(\theta \in \Tp(\Lambda, q, k)\), we write \(G \models \theta (\vec{v})\)
if \(G \models \phi(\vec{v})\) for all \(\phi(\vec{x}) \in \theta\).
Observe that \(G \models \theta(\vec{v}) \iff \tp_q^G(\vec{v}) = \theta\).
We say that a type \(\theta \in \Tp(\Lambda,q,k)\) is \emph{realizable}
if there is some \(\Lambda\)-labeled graph \(G\) and tuple \(\vec{v} \in (V(G))^k\)
such that \(\theta = \tp_q^G(\vec{v})\).
We are not particularly interested in types that are not realizable,
but it is undecidable if a type \(\theta\) is realizable,
whereas the sets \(\Tp(\Lambda, q, k)\) are decidable.
(More precisely, there is an algorithm that, given \(\Lambda, q, k\) and a set \(\theta\) of formulas,
decides if \(\theta \in \Tp(\Lambda, q, k)\).)
For a \((\Lambda, q, k)\)-type \(\theta\) and a \(\Lambda\)-labeled graph \(G\), we let
\[
  \theta(G) \deff \bigsetc{\vec{v} \in (V(G))^k}{G \models \theta(\vec{v})}.
\]
If \(\theta(G) \neq \emptyset\), we say that \(\theta\) is \emph{realizable in \(G\)}.

As for formulas, we split the variables for types into two parts,
so we consider \((\Lambda, q, k, \ell)\)-types \(\theta \subseteq \MSO(\Lambda, q, k, \ell, 0)\),
and we denote the set of all these types by \(\Tp(\Lambda, q, k, \ell)\).
For a \(\Lambda\)-labeled graph \(G\) and tuples \(\vec{v} \in (V(G))^k, \vec{w} \in (V(G))^\ell\),
we often think of \(\tp_q^G(\vec{v}, \vec{w})\) as the
\emph{\(q\)-type of \(\vec{w}\) over \(\vec{v}\) in \(G\)}.
Moreover, we let
\[
  \theta(\vec{v}, G) \deff \setc{\vec{w} \in (V(G))^\ell}
  {G \models \theta(\vec{v}, \vec{w})}.
\]
If \(\theta(\vec{v}, G) \neq \emptyset\),
we say that \(\theta\) is \emph{realizable over \(\vec{v}\) in \(G\)}.

For a vector \(\vec{k} = (k_1, \dots, k_m) \in \NN^m\) and a set \(V\),
we let \(V^{\vec{k}}\) be the set of all sequences
\(\Ca = (\vec{v}_1, \dots, \vec{v}_m)\) of tuples \(\vec{v}_i \in V^{k_i}\).
Let \(G\) be a labeled graph, \(\Ca = (\vec{v}_1, \dots, \vec{v}_m) \in V^{\vec{k}}\)
for some \(\vec{k} \in\NN^m\), and \(\vec{w} \in (V(A))^\ell\).
We define the \emph{\(q\)-type of \(\vec{w}\) over \(\Ca\) in \(G\)} to be the tuple
\[
  \tp_q^{G}(\Ca, \vec{w}) \deff \bigl(\tp_q^G(\vec{v}_1, \vec{w}), \dots,
  \tp_q^G(\vec{v}_m, \vec{w})\big).
\]
Again, we need an ``abstract'' notion of type over a sequence.
A \emph{\((\Lambda, q, \vec{k}, \ell)\)-type} for a tuple
\(\vec{k} = (k_1, \dots, k_m) \in\NN^m\) is an element of
\[
  \Tp(\Lambda, q, \vec{k}, \ell) \deff \prod_{i=1}^m \Tp(\Lambda, q, k_i, \ell).
\]
Let \(\btheta = (\theta_1, \dots, \theta_m) \in \Tp(\Lambda, q, \vec{k}, \ell)\).
For a labeled graph \(G\), a sequence \(\Ca = (\vec{v}_1, \dots, \vec{v}_m) \in (V(G))^{\vec{k}}\),
and a tuple \(\vec{w} \in (V(G))^\ell\), we write \(G \models \btheta(\Ca, \vec{w})\)
if \(G \models \theta_i(\vec{v}_i, \vec{w})\) for all \(i \in [m]\).
Note that \(G \models \btheta(\Ca, \vec{w}) \iff \tp_q^G(\Ca, \vec{w}) = \btheta\).
For a type \(\btheta \in \Tp(\Lambda, q, \vec{k}, \ell)\),
a \(\Lambda\)-labeled graph \(G\), and a sequence \(\Ca \in (V(G))^{\vec{k}}\),
we let
\[
  \btheta(\Ca, G) \deff \bigsetc{\vec{w} \in (V(G))^\ell}{G \models \theta(\Ca, \vec{w})}.
\]
If \(\btheta(\Ca, G) \neq \emptyset\),
we say that \(\btheta\) is \emph{realizable over \(\Ca\) in \(G\)}.

\subsection{VC Dimension}
For \(q, k, \ell \in \NN\), a formula \(\phi(\vec{x}, \vec{y}) \in \MSO(\Lambda, q, k, \ell,0)\),
a \(\Lambda\)-labeled graph \(G\), and a tuple \(\vec{w} \in (V(G))^\ell\),
we let
\[\phi(G, \vec{w}) \deff \bigsetc{\vec{v} \in (V(G))^k}{G \models \phi(\vec{v}, \vec{w})}.\]
For a set \(X \subseteq (V(G))^k\), we let
\[
  H_{\phi}(G,X) \deff \bigsetc{X \cap \phi(G, \vec{w})}{\vec{w} \in (V(G))^\ell}.
\]
We say that \(X\) is \emph{shattered} by \(\phi\) if \(H_\phi(G,X) = 2^X\).
The \emph{VC dimension} \(\VC(\phi, G)\) of \(\phi\) on \(G\) is the maximum \(d \in \NN\)
such that there is a set \(X \subseteq V(G)^k\) of cardinality \(\abs{X} = d\)
that is shattered by \(\phi\).
In this paper, we are only interested in finite graphs,
but for infinite \(G\), we let \(\VC(\phi, G) \deff \infty\) if the maximum does not exist.
For a class \(\CC\) of \(\Lambda\)-labeled graphs, the VC dimension of \(\phi\) over \(\CC\),
\(\VC(\phi, \CC)\), is the least \(d\) such that \(\VC(\phi, G) \leq d\) for all \(G \in \CC\)
if such a \(d\) exists, and \(\infty\) otherwise.

\begin{lemma}[{\cite[Theorem~17]{grohe_learnability_2004}}]\label{lem:gt}
  There is a function \(g \colon \NN^5 \to \NN\) such that the following holds.
  Let \(\Lambda\) be a label set, let \(\CC\) be the class of all \(\Lambda\)-graphs
  of clique-width at most \(c\), and let \(q, k, \ell \in \NN\).
  Then \(\VC(\phi, \CC) \leq g(c, \abs{\Lambda}, q, k, \ell)\)
  for all \(\phi \in \MSO(\Lambda, q, k, \ell, 0)\).
\end{lemma}

\subsection{Parameterized Complexity}
\label{subsec-para-complex}
A \emph{parameterization} $\kappa$ is a function mapping the input $x$ of a problem
to a natural number $\kappa(x) \in \mathbb{N}$.
An algorithm $\mathbb{A}$ is an \emph{fpt algorithm with respect to $\kappa$}
if there is a computable function $f \colon \NN \to \NN$
and a polynomial $p$ such that for every input $x$ the running time of $\mathbb{A}$
is at most $f(\kappa(x)) \cdot p(\abs{x})$.

A \emph{parameterized problem} is a tuple $(Q,\kappa)$.
We say $(Q,\kappa) \in \FPT$ or $(Q,\kappa)$ is \emph{fixed-parameter tractable}
if there is an fpt algorithm with respect to $\kappa$ for $Q$,
and we say \((Q, \kappa)\) is \emph{fixed-parameter linear}
if the polynomial in the running time of the fpt algorithm is linear.
We say \((Q,\kappa) \in \text{para-NP}\) if there is a nondeterministic fpt algorithm
with respect to \(\kappa\) for \(Q\).
If the parameterization is clear from the context, then we omit it.

For two parameterized problems \((Q,\kappa), (Q', \kappa')\),
an \emph{fpt Turing reduction} from \((Q,\kappa)\) to \((Q',\kappa')\)
is an algorithm \(\mathbb{A}\) with oracle access to \(Q'\)
such that \(\mathbb{A}\) decides \(Q\),
\(\mathbb{A}\) is an fpt algorithm with respect to \(\kappa\),
and there is a computable function \(g \colon \NN \to \NN\)
such that on input \(x\),
\(\kappa'(x') \leq g\bigl((\kappa(x)\bigr)\) for all oracle queries with oracle input \(x'\).

For additional background on parameterized complexity, we refer to \cite{flum_parameterized_2006}.
 \section{Tractability for One-Dimensional Training Data on Well-Behaved Classes}
\label{sec:tract}

We start by formalizing the parameterized version of the problem \MSOLearn described in the introduction.
For a training sequence $S$, a graph $G$, and a hypothesis \(h_{\phi, \vec{w}}\),
we say \emph{\(h_{\phi, \bar{w}}\) is consistent with $S$ on $G$}
if for every positive example \((\vec v, +) \in S\), we have \(G \models \phi(\vec v, \vec w)\),
and for every negative example \((\vec v, -) \in S\), we have \(G \not\models \phi(\vec v, \vec w)\).

\problemDefPara{\MSOLearn}
    {$\Lambda$-labeled graph $G$, $q, k, \ell \in \mathbb{N}$, training sequence $S \in (V(G)^{k} \times \{+,-\})^m$}
    {$\kappa \deff |\Lambda|+q+k+\ell$}
    {Return a hypothesis $h_{\phi,\bar{w}}$ consisting of
    \begin{itemize}
        \item a formula $\phi \in \MSO(\Lambda, q, k, \ell,0)$ and
        \item a parameter setting $\bar{w} \in V(G)^{\ell}$
    \end{itemize}
    such that $h_{\phi,\bar{w}}$ is consistent with the training sequence $S$ on $G$, if such a hypothesis exists. Reject if there is no consistent hypothesis.}

The problem \uMSOLearn refers to the \(1\)-dimen\-sional version of the problem \MSOLearn
where the arity \(k\) of the training examples is $1$.
The tractability results for \uMSOLearn are significantly more straightforward than those for the higher-dimensional problem.
This is due to the fact that the full training sequence can be encoded into the graph by only adding two new labels as follows.

\begin{lemma}
  \label{Lemma:codeExamplesIntoStruct}
  Given a \(\Lambda\)-labeled graph $G$,
  an \(\MSO\) formula $\phi(x, \bar{y})$ of quantifier rank $q_\phi$,
  and a training sequence $S \in \bigl(V(G) \times \{+,-\}\bigr)^m$, there is
  \begin{itemize}
      \item a label set $\Lambda'$ of size $|\Lambda| + 2$,
      \item a $\Lambda'$-labeled graph $G_S$,
      \item and an \(\MSO\) formula $\phi'(\bar{y})$ of quantifier rank $q_\phi+1$
  \end{itemize}
  such that for all $\bar{w} \in V(G)^\ell$,
  we have $G_S \models \phi'(\bar{w})$ if and only if
  $h_{\phi,\bar{w}}$ is consistent with $S$ on $G$.
\end{lemma}

\begin{proof}
Let \(\Lambda' \deff \Lambda \uplus \set{P,N}\) for two new labels \(P, N\).
Furthermore, let \(G_S\) be the \(\Lambda'\)-structure with the same vertex set as \(G\), and the labels
\(R(G_S) = R(G)\) for $R \in \Lambda$,
\(P(G_S) = \setc{v \in V(G)}{(v,+) \in S}\), and
\(N(G_S) = \setc{v \in V(G)}{(v,-) \in S}\).

Based on the \(\MSO(\Lambda, q_\phi, 1, \ell, 0)\) formula \(\phi(x;\vec{y})\),
we define the \(\MSO(\Lambda', q_\phi+1, 0, \ell, 0)\) formula \(\phi'(\bar{y})\) as
\begin{align*}
    \phi' (y_1, \dots, y_\ell) \deff
    \forall x \Bigl(&\bigl(P(x) \rightarrow \phi(x, y_1, \dots, y_\ell)\bigr) \land\\
     &\bigl(N(x) \rightarrow \neg \phi(x, y_1, \dots, y_\ell)\bigr)\Bigr).
\end{align*}
Then, we have $G_S \models \phi'(\bar{w})$
if and only if $h_{\phi,\bar{w}}$ is consistent with the training sequence $S$ on $G$.
\end{proof}

As discussed in \cref{sec-prelim-formulas},
there is a function \(f \colon \NN^4 \to \NN\) such that
\(\abs{\MSO(\Lambda, q,k,\ell, 0)}  \leq f(\abs\Lambda, q,k,\ell)\).
Therefore, to solve \uMSOLearn, we can iterate over all formulas
\(\phi \in \MSO(\Lambda, q,k,\ell, 0)\) and focus on finding a parameter setting
\(\vec w \in V(G)^\ell\) such that \(h_{\phi,\vec w}\) is consistent with \(S\) on \(G\).
Moreover, if the model-checking problem on a graph class with additional labels is tractable,
then finding a consistent parameter setting is tractable as well
by performing model checking on the graph with the encoded training sequence.

\begin{lemma}
  \label{lem:unary-tractability-general}
  Let \(\CC\) be a class of labeled graphs,
  let \(\CC_i\) be the class of all extensions of graphs from \(\CC\) by $i$ additional labels for all \(i \in \NN\),
  let \(f \colon \NN \to \NN\) be a function,
  and let \(c \in \NN\) such that the \(\MSO\)-model-checking problem
  on \(\CC_i\) can be solved in time \(f(\abs{\phi}) \cdot \abs{V(G)}^c\) for all \(i \in \NN\),
  where \(\phi\) is the \(\MSO\) sentence and
  \(G \in \mathcal{C}_i\) is the labeled graph given as input.
There is a function \(g \colon \NN^3 \to \NN\) such that
  \(\uMSOLearn\) can be solved on \(\mathcal{C}\)
  in time \(g(\abs{\Lambda}, q, \ell) \cdot \abs{V(G)}^{c+1}\).
\end{lemma}

\begin{proof}
We can iterate over all formulas \(\phi \in \MSO(\Lambda, q, 1, \ell, 0)\)
and encode the training sequence into the graph
using \cref{Lemma:codeExamplesIntoStruct},
resulting in a \(\Lambda'\)-labeled graph \(G_S\)
and a formula \(\phi'\).
There is a consistent parameter setting \(\vec w \in (V(G))^\ell\) if and only if
it holds that \(G_S \models \exists y_1, \dots \exists y_\ell \phi'(y_1, \dots, y_\ell)\),
which can be checked using an algorithm for the \(\MSO\)-model-checking problem.

We enforce assigning a specific vertex $v_1 \in V(G)$ to the
variable $y_1$ by building a new graph \(G_{S,v_1}\) and encoding the
parameter choice as a new label \(P_{y_1}\) with \(P_{y_1}(G_{S,v_1})=\{v_1\}\).
For $\vec w = (w_1, \dots ,w_\ell) \in V(G)^\ell$ where $w_1 = v_1$,
we have $G_S \models \phi'(\vec w)$ if and only if
\[G_{S,v_1} \models \exists y_1 \dots \exists y_\ell \bigl(P_{y_1}(y_1)
\land \phi'(y_1, y_2 \dots, y_\ell)\bigr).\]
We then perform model checking on $G_{S,v_1}$ to see if
we have chosen a vertex $v_1$ for $y_1$ that can be extended into a
consistent parameter setting.
Once such a vertex is found, we can continue
with the next variable. In total, we add \(\ell+2\) new labels to the graph.
In this way, a consistent
parameter setting can be computed in at most $\abs{V(G)}\cdot \ell$ rounds of model checking.
\end{proof}

If the input graph is given in the form of a $c$-expression,
then the \(\MSO\) model-checking problem is fixed-parameter linear on
classes of bounded clique-width \cite{courcelle_linear_1999}.
Therefore, \cref{lem:unary-tractability-general} implies that
there is a function \(g \colon \NN^4 \to \NN\) such that \(\uMSOLearn\)
can be solved in time \(g(\cw(G), \abs{\Lambda}, q, \ell) \cdot \abs{G}^2\).

\Cref{thm:1d-tract-main} improves this bound for classes of graphs of bounded clique-width even further,
showing that the problem \(\uMSOLearn\) can be solved in time linear in the size of the graph.

\oneDConsistentLearning*

Since graphs of tree-width \(c_t\) have a clique-width of
at most \(3\cdot2^{c_t-1}\) \cite{corneil_relationship_tw_cw},
\cref{thm:1d-tract-main} implies that for classes of graphs of bounded tree-width,
\(\uMSOLearn\) is fixed-parameter linear as well.
Moreover, although all background structures we consider in this paper are labeled graphs,
we remark that the result for classes of bounded tree-width also holds on arbitrary relational structures
and a corresponding version of \(\uMSOLearn\).

The remainder of this section is dedicated to the proof of \cref{thm:1d-tract-main}.
As argued above,
since the number of possible formulas to check can be bounded in terms of
\(\abs{\Lambda}\), \(q\), and \(\ell\),
it suffices to show how to find a parameter setting \(\vec{w}\)
for a given formula \(\phi\) such that the hypothesis \(h_{\phi, \vec{w}}\)
is consistent with the training sequence \(S\),
or reject if there is no such parameter setting.

\subsection{Preprocessing the Input Graph}
\label{subsec:clique-pre}

Let \(G\) be a \(\Lambda\)-labeled graph
and let \(S \in \bigl(V(G) \times \set{+,1}\bigr)^m\)
be the training sequence given as input for \uMSOLearn.
We assume that we are given \(G\)
in form of a \(c\)-expression \(\Xi\).
Let \(\phi \in \MSO(\Lambda, q, 1, \ell, 0)\)
be the formula for which we want to find a consistent parameter setting.

As described in \cref{Lemma:codeExamplesIntoStruct},
we encode the training sequence \(S\) into \(G\).
This results in a \(\Lambda'\)-labeled graph \(G_S\)
(with \(\abs{\Lambda'} = \abs{\Lambda} + 2\)),
a corresponding \((c+2)\)-expression \(\Xi_S\),
and an \(\MSO(\Lambda', q+1, 0, \ell, 0)\) formula \(\phi_2(\bar{y})\)
such that
\[
    G_S \models \phi_2(\vec w) \iff \text{\(h_{\phi,\vec w}\) is consistent with $S$ on $G$}.
\]

Further, using the following lemma,
we transform \(\phi_2\) into a formula \(\phi_3\)
that has free set variables instead of free individual variables.
This allows us to directly use the results from \cite{courcelle_linear_1999}.

\begin{restatable}{lemma}{removeFOquantifier}
  \label{lemma:remove-FO-quantifier}
  For a \(\Lambda'\)-labeled graph $G_S$
  and an \(\MSO(\Lambda', q+1, 0, \ell, 0)\) formula \(\phi_2(y_1, \dots, y_\ell)\),
  there is an \(\MSO(\Lambda', q+\ell+3, 0, 0, \ell)\) formula $\phi_3(X_1, \dots, X_\ell)$
  such that for all $v_1, \dots, v_\ell \in V(G_S)$, we have that
  \[G_S \models \phi_2(v_1, \dots, v_\ell) \iff
  G_S \models \phi_3(\set{v_1}, \dots, \set{v_\ell}).\]
  Moreover, for all \(V_1, \dots, V_\ell \subseteq V(G_S)\),
  we have
  \(G_S \not\models \phi_3(V_1, \dots, V_\ell)\) if
  \(\abs{V_i} \neq 1\) for any \(i \in [\ell]\).
\end{restatable}
\begin{proof}
  Let
  \[Sing(X_i,y_i) \deff X_i(y_i) \land \forall x \forall y \Bigl(\bigl(X_i(x)
  \land X_i(y)\bigr) \rightarrow x=y\Bigr)\]
  be the formula that enforces that the set variable $X_i$ has exactly one element, which is $y_i$.
  Now we can bind all free individual variables and replace them by free set variables
  by setting
  \[
    \phi_3(X_1, \dots, X_\ell) \deff \exists y_1 \dots \exists y_\ell
    \Bigl(\Land_{i \in [\ell]} Sing(X_i, y_i) \land \phi_2(y_1, \dots, y_\ell)\Bigr).
  \]
  This formula satisfies the requirements stated in the lemma.
\end{proof}

\subsection{Computing the Parameters}
\label{subsec:clique-formulas}
For every \(\MSO(\Lambda')\) formula \(\psi(X_1, \dots, X_\ell)\) and every \(\Lambda'\)-labeled graph \(G\),
we let
\begin{equation*}
    \psi(G) \deff \bigsetc{(V_1, \dots, V_\ell)}{G \models \psi(V_1, \dots, V_\ell)}
\end{equation*}
be the set of all parameter settings for which $\psi$ holds in $G$.
Note that any $(V_1,...,V_\ell) \in \phi_3(G_S)$ yields a consistent parameter setting for $\phi$.

Computing all possible parameter settings is intractable,
as the size of $\phi_3(G_S)$ might be up to $\abs{V(G_S)}^{\ell}$.
Instead, we follow the construction from \cite{courcelle_linear_1999} to compute a subset of $\phi_3(G_S)$
that is empty if and only if $\phi_3(G_S)$ is empty.
Let \(q' \deff q + \ell + 3\).
We use the following three
results to recursively compute the subset.

\begin{lemma}[\cite{courcelle_linear_1999}]
\label{lemma:sat-unary-trans}
  For every \(P, Q \in \Lambda'\) with \(P \neq Q\),
  for every operation \(f \in \set{\rho_{P, Q}, \eta_{P, Q}}\),
  and for every $\psi \in \MSO(\Lambda',q',0,0,\ell)$,
  there is a formula $\psi' \in \MSO(\Lambda',q',0,0,\ell)$ of the same quantifier rank
  such that for all \(\Lambda'\)-labeled graphs \(G'\),
  \begin{equation*}
      \psi(f(G')) = \psi'(G').
  \end{equation*}
  Moreover, \(\psi'\) can be computed from \(\psi\).
\end{lemma}

It is easy to see that this also holds for the operation \(\delta_P\).

\begin{lemma}
\label{lemma:sat-delete}
  For every \(P \in \Lambda'\)
  and $\psi \in \MSO(\Lambda',q',0,0,\ell)$,
  there is a formula $\psi' \in \MSO(\Lambda',q',0,0,\ell)$ of the same quantifier rank
  such that for all \(\Lambda'\)-labeled graphs \(G'\),
  \begin{equation*}
      \psi(\delta_P(G')) = \psi'(G').
  \end{equation*}
  Moreover, \(\psi'\) can be computed from \(\psi\) by replacing all occurrences of \(P(x)\) in \(\psi\) by \(\bot\).
\end{lemma}

\begin{lemma}[\cite{makowsky_fefermanvaught_2004}, Feferman--Vaught for \MSO]
\label{lemma:sat-binary-trans}
For every $\psi \in \MSO(\Lambda',q',0,0,\ell)$,
there are $\MSO(\Lambda',q',0,0,\ell)$ formulas
$\theta_1, \dots, \theta_m$ and $\chi_1, \dots, \chi_m$
such that all formulas have the same free variables as $\psi$
and have quantifier rank no larger than the quantifier rank of $\psi$,
and for every two $\Lambda'$-labeled graphs $G_1$ and $G_2$
such that \(G_1\) and \(G_2\) are disjoint,
\begin{equation*}
    \psi(G_1 \uplus G_2) = \bigcup_{i=1}^m \theta_i(G_1) \boxtimes \chi_i(G_2),
\end{equation*}
where
\begin{align*}
  A \boxtimes B \deff
  &\bigl\{(A_1 \cup B_1, \dots, A_n \cup B_n) \bigmid \\
  &\qquad (A_1, \dots, A_n) \in A, (B_1, \dots, B_n) \in B\bigr\}.
\end{align*}
Moreover, $\theta_1, \dots, \theta_m$ and $\chi_1, \dots, \chi_m$
can be computed from \(\psi\).
\end{lemma}

In the following, we view the $(c+2)$-expression \(\Xi_S\) as a tree of operations which we refer to by $\T_{\Xi_S}$.
We start with assigning the formula $\phi_3$ to the root of the tree $\T_{\Xi_S}$.
From there, we assign formulas to each inner node of the tree.
If a node represents an operation \(\rho_{P,Q}\), \(\eta_{P,Q}\) or \(\delta_P\),
then we apply \cref{lemma:sat-unary-trans} or \cref{lemma:sat-delete}
respectively to every formula assigned to the node and assign the
resulting formulas to the child node.
If the node represents an operation \(\uplus\),
then we apply \cref{lemma:sat-binary-trans} to every formula assigned to the node
and assign the formulas \(\theta_1, \dots, \theta_m\) to the first child
and \(\chi_1, \dots, \chi_m\) to the second child.
Note that the quantifier rank, the number of free variables,
and the vocabulary stay the same in each step.
We can bound the number of formulas assigned no each node
by \(\abs{\MSO(\Lambda',q',0,0,\ell)}\),
which only depends on \(\abs{\Lambda}\), \(q\), and \(\ell\).

After this step, we have a list of formulas $\Psi$ assigned to each node of $\T_{\Xi_S}$.
Starting at the leafs, which correspond to base graphs $G_v$ that only have a single node $v$,
we compute the sets of consistent parameter settings \(\psi(G_v)\) for each formula \(\psi \in \Psi\).
On a graph with a single node, $\psi(G_v)$ can be computed by a simple brute-force algorithm
in time $\bigO\bigl(2^\ell \cdot 2^{q'} \cdot \abs{\phi}\bigr)$.
We then construct a parameter setting bottom up from the leafs,
following \cref{lemma:sat-unary-trans,lemma:sat-delete,lemma:sat-binary-trans}.
\begin{restatable}{lemma}{lemmaSatSingle}
\label{lamma_save_single_form}
With the above-described procedure,
we find a consistent parameter setting \(\bar{V} \in \phi_3(G_S)\)
if and only if we find a consistent parameter setting in an analogous procedure where,
in each node of \(\T_{\Xi_S}\), corresponding to a labeled graph $G'$,
we only keep one parameter setting \(\bar{V}' \in \phi'(G')\) for every formula \(\phi'\)
assigned to the node.
\end{restatable}

\begin{proof}
We show that if we can find a parameter setting in the root using all possible parameter settings,
we can also find a parameter setting in the root using only a single representative setting per node and formula.
Unary operations $\rho_{P,Q}$ and $\eta_{P,Q}$ do not change the number of parameter settings in a node,
and the initialization operation is executed only on leafs.
Now consider the case of $G_1 \uplus G_2$ with formula $\psi$
and with corresponding formulas $\theta_1, \dots, \theta_m$ and $\chi_1, \dots, \chi_m$ in the children.
For all \(i \in [m]\), let \(\vec{g}_{1,i}\) be the single representative setting saved for \(\theta_i\),
where \(\vec{g}_{1,i} = \emptyset\) if \(\theta_i(G_1) = \emptyset\),
and \(\vec{g}_{1,i} \in \theta_i(G_1)\) else.
Analogously, let \(\vec{g}_{2,i}\) be the single representative setting saved for \(\chi_i\) in \(G_2\).

In the above-described procedure, we find a consistent parameter setting
if and only if \(\psi(G_1 \uplus G_2) \neq \emptyset\).
By \cref{lemma:sat-binary-trans},
this holds if and only if
\(\theta_i(G_1) \boxtimes \chi_i(G_2) \neq \emptyset\),
or, equivalently, \(\theta_i(G_1) \neq \emptyset\)
and \(\chi_i(G_2) \neq \emptyset\),
for some \(i \in [m]\).
This holds if and only if \(\bar{g}_{1,i} \neq \emptyset\) and \(\bar{g}_{2,i} \neq \emptyset\)
for some \(i \in [m]\).
These are precisely the cases in which we find a consistent parameter setting
if we only keep one parameter setting for every formula in every node.
\end{proof}
Now, we can reconstruct a parameter setting for the root,
following \cref{lemma:sat-unary-trans,lemma:sat-delete,lemma:sat-binary-trans}.
Since there is only a constant number of formulas assigned to each node, say \(c = g(\abs\Lambda, q, \ell)\) for some $g \colon \NN^3 \to \NN$,
and we only save at most one parameter setting per formula and node,
we compute at most \(c\) settings for each node using \cref{lemma:sat-unary-trans} or \cref{lemma:sat-delete},
and at most \(c^2\) settings for each node using \cref{lemma:sat-binary-trans},
which we can then again reduce to \(c\) settings.
Hence, all in all, there is a function \(f \colon \NN^4 \to \NN\)
such that a consistent hypothesis $h_{\phi,\vec w}$ is found in time
\(\bigO\bigl(\abs{V(G)}\cdot f(c, \abs\Lambda, q, \ell)\bigr)\).
This concludes the proof of \cref{thm:1d-tract-main}.
 \section{Hardness for One-Dimensional Training Data}
\label{sec:hardness}

Previously, we restricted the input graph of the \(\MSO\)-learning problem
to certain well-behaved classes.
Now, we consider the problem \MSOLearn without any restrictions.
Van Bergerem, Grohe, and Ritzert showed
in~\cite{van_bergerem_parameterized_2022} that there is a close relation
between first-order model checking (\FOMc)
and learning first-order formulas.
The fpt-reduction in~\cite{van_bergerem_parameterized_2022}
from model checking to learning yields \(\AWstar\)-hardness
for learning first-order formulas on classes of structures that are not nowhere dense.
It is simple to show (and not surprising) that \MSOLearn is at least as hard as \FOMc.
The more interesting question is whether \MSOLearn is at least as hard
as the model-checking problem for \(\MSO\) sentences (\MSOMc),
which is defined as follows.

\problemDefPara{\MSOMc}
{\(\Lambda\)-labeled graph \(G\), \(\MSO(\Lambda)\) sentence \(\phi\)}
{\(\abs{\phi}\)}
{Decide whether \(G \models \phi\) holds.}

We give a positive answer,
which even holds for the \(\MSO\)-learning problem with only one-dimensional training data
where we restrict the training sequence to contain at most two training examples.

\begin{restatable}{lemma}{hardness}
  \label{lemma:reduction}
  The model-checking problem \(\MSOMc\) is fpt Turing reducible to
  \(\problemFont{1D-MSO}\)-\(\problemFont{Consistent-Learn}\)
  where we restrict the training sequence \(S\) given as input
  to have length at most \(2\).
\end{restatable}

\(\MSOMc\) is para-NP-hard under fpt Turing reductions
as even for some fixed sentence \(\phi\),
the corresponding model-checking problem can be NP-hard
(for example for a formula defining \problemFont{3-Colorability},
see \cite{flum_parameterized_2006} for details).
Hence, \cref{lemma:reduction} proves \cref{theorem:reduction}.

\begin{proof}[Proof of \cref{lemma:reduction}]
  We describe an fpt algorithm with access to a
\(\problemFont{1D-MSO-Consistent-}\) \(\problemFont{Learn}\)
  oracle that solves \(\MSOMc\).
  Let \(G\) be a \(\Lambda\)-labeled graph,
  and let \(\phi\) be an \(\MSO(\Lambda)\) sentence.
  Let \(n \deff \abs{V(G)}\).
  W.l.o.g., we may assume that \(\Lambda\) only contains labels that appear in \(\phi\) and hence,
  we have \(\abs{\Lambda} \leq \abs{\phi}\).
  We decide whether \(G \models \phi\) holds recursively
  by decomposing the input formula.
  While handling negation and Boolean connectives is easy,
  the crucial part of the computation is handling quantification.
  Thus, we assume that \(\phi = \exists x \psi\) or \(\phi = \exists X \psi\)
  for some \(\MSO\) formula \(\psi\).
  For both types of quantifiers, we use the \uMSOLearn oracle to identify
  a small set of candidate vertices or sets such that \(\psi\) holds for any vertex or set
  if and only if it holds for any of the identified candidates.
  Then, since the number of candidates will only depend on \(\abs{\psi}\),
  we can check recursively whether \(\psi\) holds for any of them
  and thereby decide \MSOMc with an fpt algorithm.

  First, assume that \(\phi = \exists x \psi\).
  For every pair of distinct vertices \(v, v' \in V(G)\),
  we call the \uMSOLearn oracle with input graph \(G\),
  \(q = \qr(\psi)\), \(\ell = 0\),
  and training sequence \(\TS = \bigl((v, +), (v', -)\bigr)\).
  Since \(\ell = 0\), the learning algorithm may not use any vertices as parameters
  in the hypothesis.
  Thus, the oracle returns a hypothesis if and only if \(\type{q}{G}{v} \neq \type{q}{G}{v'}\).
  Note that the oracle answers induce a partition on \(V(G)\)
  where two vertices \(v\) and \(v'\) are in the same class
  if and only if the oracle does not return a hypothesis
  on the input as specified above.
  The number of classes is the number of vertex \(q\)-types in \(G\),
  which only depends on \(q\) and \(\abs{\Lambda}\) and can thus be
  bounded in terms of \(\abs{\phi}\).
  Finally, since \(G \models \psi(v)\) if and only if \(G \models \psi(v')\)
  for all vertices \(v\) and \(v'\) in the same class,
  it suffices to check a single candidate per class.

  Let \(v \in V(G)\) be the candidate for which we want to check
  whether \(G \models \psi(v)\) holds.
  To be able to run our algorithm for \MSOMc recursively,
  we encode our choice via labels in the graph and then turn the formula \(\psi\)
  into a sentence.
  For this, let \(\Lambda' \deff \Lambda \uplus \set{I_v, N_v}\).
  For \(v \in V(G)\), let \(G_v\) be the \(\Lambda'\)-labeled graph with vertex set \(V(G_v) = V(G)\)
  and labels \(R(G_v) \deff R(G)\) for all \(R \in \Lambda\),
  \(I_v(G_v) \deff \set{v}\),
  and \(N_v(G_v) \deff \setc{w \in V(G)}{(v,w) \in E(G)}\).
  W.l.o.g., we assume that the formula \(\psi(x)\) does not contain atoms of the form \(x = x\).
  The formula can then be transformed into a sentence \(\psi_v\)
  by replacing every atom of the form \(y = x\) or \(x = y\) by \(I_v(y)\),
  and by replacing atoms of the form \(E(x,y)\) or \(E(y,x)\) by \(N_v(y)\).
  We have \(G \models \psi(v)\) if and only if \(G_v \models \psi_v\).
  Furthermore, \(\abs{\Lambda'} = \abs{\Lambda} + 2\),
  and \(G_v\) and \(\psi_v\) can be computed by an fpt algorithm.

  Note that for now, considering only first-order quantifiers,
  the depth and the degree of the recursion tree can be bounded in terms of \(\abs{\phi}\).
  Moreover, in all oracle queries, the vertex set of the input graph is \(V(G)\),
  we only use training sequences of length \(2\),
  and the parameter of the problem \uMSOLearn can be bounded in terms of \(\abs{\phi}\).

  In the remainder of this proof, we describe how to handle second-order quantifiers.
  Assume that \(\phi = \exists X \psi\).
  Similarly to the first case, we want to find a small set of candidate sets for \(X\).
  However, in contrast to the first-order quantifier,
  querying the oracle for every possible choice would lead to a running time
  that is exponential in the size of the input graph.
  To overcome this issue, intuitively, we proceed as follows.
  We go through the possible choices in \(n\) rounds,
  where in round \(i\), we only consider candidate sets of size \(i\).
  Overall, we compute \(n\) families \(V_1, \dots, V_n\) of candidates,
  where every single family \(V_i\) is small, that is,
  the size of \(V_i\) only depends on \(\abs{\phi}\), say \(f(\abs{\phi})\).
  In the first round, we compute \(V_1\) by going through all sets of size \(1\).
  Based on answers of the \uMSOLearn oracle,
  we keep only \(f(\abs{\phi})\) many candidates.
  Then, in round \(i+1\) for \(i>0\),
  we consider all extensions of the sets in \(V_i\) by a single element,
  which are \(f(\abs{\phi}) \cdot n\) many possible choices to consider.
  Again, using the oracle, we keep only \(f(\abs{\phi})\) many candidates
  and call the resulting family of candidates \(V_{i+1}\).
  After \(n\) rounds, all in all, we have \(f(\abs{\phi}) \cdot n\)
  many candidates.
  In a last step, we then also compare candidates of different sizes,
  ending up with at most \(f(\abs{\phi})\) candidates in total.
  Similarly to the first-order quantifier, for every candidate,
  we can encode the candidate via labels in the graph
  and run our algorithm for \MSOMc recursively.
  In the following, we formalize every single step of the described procedure.

  First, we describe how to find the small number of candidates in each round.
  For a set \(C \subseteq V(G)\), let the \(q\)-type \(\tp_q^G(C)\) of \(C\) in \(G\) be the set
  of all \(\MSO(\Lambda, q, 0, 1)\) formulas \(\alpha(X)\) with \(G \models \alpha(C)\).
  Now, for every \(\MSO(\Lambda, q, 0, 1)\) formula \(\gamma(X)\), consider the \(\MSO(\Lambda, q+2, 0, 1)\) formula
  \[\gamma'(X) = \exists Y \bigl(\forall x (X(x) \rightarrow Y(x)) \land \gamma(Y)\bigr).\]
  This formula checks whether any superset of \(X\) (including \(X\) itself) satisfies \(\gamma\).
  Thus, if two sets \(C_1, C_2\) have the same \((q+2)\)-type,
  then for every superset \(C_1' \supseteq C_1\),
  there is a superset \(C_2' \supseteq C_2\) of the same \(q\)-type.
  Hence, when computing the families \(V_i\), it suffices to keep only one of the two sets.

  To check for two sets \(C_1, C_2 \subseteq V(G)\) whether they have the same \((q+2)\)-type,
  we use the \uMSOLearn oracle.
  As the input graph, we choose the \(\Lambda'\)-labeled graph \(G'\)
  with \(\Lambda' \deff \Lambda \uplus \set{C}\),
  vertex set \(V(G') \deff \setc{(v,1), (v,2)}{v \in V(G)} \uplus \set{w_1, w_2}\),
  edge set
  \begin{align*}
    E(G') \deff
    &\bigsetc{\bigl((v, i), (v', i)\bigr)}{(v,v') \in E(G), i \in \set{1,2}}\\
    &\cup \bigsetc{\bigl((v, i), w_i\bigr), \bigl(w_i, (v, i)\bigr)}{v \in V(G), i \in \set{1,2}},
  \end{align*}
  and labels
  \(C(G') \deff \bigsetc{(v,1)}{v \in C_1} \cup \bigsetc{(v,2)}{v \in C_2}\)
  and \(R(G') \deff \bigsetc{(v,1), (v,2)}{v \in R(G)}\) for all \(R \in \Lambda\).

  Intuitively, \(G'\) contains two copies of \(G\).
  For each copy, there is vertex \(w_i\) that is connected to every vertex in the copy.
  Furthermore, the label \(C\) encodes the set \(C_1\) in the first copy
  and \(C_2\) in the second copy. The construction is depicted in \cref{fig:hard}.

  \begin{figure}
      \centering
      \begin{tikzpicture}
        \node[vertex] (w1) [label=above:$w_1$] {};
\node[vertex, right of=w1, node distance=12em] (w2) [label=above:$w_2$] {};

\foreach \i in {1,2} {
  \node[below of=w\i, node distance=7.5em] (g\i) {};
  \foreach \edge in {0,...,10} {
    \draw[edge, gray] (w\i) to [bend right=20-4*\edge] ($(g\i)+(-4.5em + \edge*.9em,0)$);
  }
  \node[draw, semithick, dashed, circle, fill=black!2, inner sep=3.5em, below of=g\i, node distance=0] (graph\i) [label=below left:$G$] {};
}

\node[ellipse, fill=ibm-magenta!50, minimum width=4.5em, minimum height=2.5em] (c1) at ($(g1) + (-1em, -1em)$) {\(C_1\)};
\node[ellipse, fill=ibm-magenta!50, minimum width=4.5em, minimum height=2.5em] (c2) at ($(g2) + (1em, 1em)$) {\(C_2\)};

\draw[dashed, gray, rounded corners=.5em] ($(g1) + (-6em, -6em)$) rectangle ($(g2) + (6em, 10em)$);
\node[color=gray] at ($(g1) + (-7em, -5.5em)$) {\(G'\)};
       \end{tikzpicture}
      \caption{Input graph $G'$ of the \uMSOLearn oracle.}
      \label{fig:hard}
  \end{figure}

  \begin{claim}
    \label{claim:hardness-second-order}
    For an algorithm solving the problem \uMSOLearn,
    given \(G'\) as the input graph,
    \(2q+6\) as the quantifier rank,
    \(\ell = 0\),
    and \(\TS= \bigl((w_1, +), (w_2, -)\bigr)\) as the training sequence,
    the following holds.
    \begin{enumerate}
      \item\label{item:hardness-second-order-reject-to-type}
        If the algorithm rejects, then \(\type{q+2}{G}{C_1} = \type{q+2}{G}{C_2}\).
      \item\label{item:hardness-second-order-type-to-reject}
        If \(\type{2q+6}{G}{C_1} = \type{2q+6}{G}{C_2}\),
        then the algorithm rejects.
    \end{enumerate}
  \end{claim}
  \begin{claimproof}
    We prove \eqref{item:hardness-second-order-reject-to-type} by showing that
    the algorithm returns a hypothesis for all \(C_1, C_2\) of different \((q+2)\)-types.
    Let \(C_1, C_2 \subseteq V(G)\) with
    \(\type{q+2}{G}{C_1} \neq \type{q+2}{G}{C_2}\),
    let \(\alpha(X) \in \type{q+2}{G}{C_1} \setminus \type{q+2}{G}{C_2}\),
    and
    \[\beta(x) \deff \exists X \Bigl(\forall y \bigl(Xy \leftrightarrow (Cy \land Exy)\bigr)
      \land \alpha'(x, X)\Bigr),\]
    where \(\alpha'(x, X)\) is computed recursively from \(\alpha(X)\)
    by replacing every subformula of the form \(\exists Y \xi\) by
    \(\exists Y \bigl(\forall y (Yy \rightarrow Exy) \land \xi'\bigr)\)
    and replacing every subformula of the form \(\exists y\, \xi\) by
    \(\exists y (Exy \land \xi')\).
    It holds that \(G \models \alpha(C_i)\)
    if and only if \(G' \models \beta(w_i)\) for \(i \in [2]\).
    Furthermore, \(\qr(\beta) \leq 2 \qr(\alpha) + 2 = 2q+6.\)
    Hence, an algorithm solving \uMSOLearn will return a hypothesis
    (for example \(h_\beta\)).

    To prove \eqref{item:hardness-second-order-type-to-reject},
    we assume that \(\type{2q+6}{G}{C_1} = \type{2q+6}{G}{C_2}\).
    It suffices to show that
    \(\type{2q+6}{G'}{w_1} = \type{2q+6}{G'}{w_2}\),
    since an algorithm solving \(\uMSOLearn\) will not be able to return
    a consistent hypothesis in that case and reject.
    Let \(V_i \deff \bigsetc{(v, i)}{v \in V(G)}\)
    and \(V'_i \deff V_i \cup \set{w_i}\) for \(i \in [2]\).
    Then, for every \(\MSO(\Lambda', 2q+6, 0, 0)\) sentence \(\chi\),
    we have \(G'[V_1] \models \chi\) if and only if \(G'[V_2] \models \chi\).
    To see this, assume that there is an \(\MSO(\Lambda', 2q+6, 0, 0)\) sentence \(\chi\)
    with \(G'[V_1] \models \chi\) and \(G'[V_2] \not\models \chi\).
    Then \(\chi'(C) \deff \chi\) is an \(\MSO(\Lambda, 2q+6, 0, 1)\) formula that is contained in
    \(\type{2q+6}{G}{C_1}\) and not contained in \(\type{2q+6}{G}{C_2}\),
    which contradicts \(\type{2q+6}{G}{C_1} = \type{2q+6}{G}{C_2}\).

    From this result on \(G'[V_1]\) and \(G'[V_2]\),
    for example by considering the corresponding Ehrenfeucht-–Fraïssé games,
    we obtain that
    \(G'[V'_1] \models \chi\) if and only if \(G'[V'_2] \models \chi\)
    and \(\type{2q+6}{G'[V'_1]}{w_1} = \type{2q+6}{G'[V'_2]}{w_2}\).
    Since \(G'\) is the disjoint union of \(G'[V'_1]\) and \(G'[V'_2]\),
    by the Feferman--Vaught Theorem for \(\MSO\) \cite{makowsky_fefermanvaught_2004},
    this implies that \(\type{2q+6}{G'}{w_1} = \type{2q+6}{G'}{w_2}\).
  \end{claimproof}

  To compute the families of candidate sets \(V_1, \dots, V_n\),
  in the first round,
  we set \(V_1 = \bigsetc{\set{v}}{v \in V(G)}\).
  Then, for every pair \(C_1, C_2 \in V_1\), we query the oracle as described above.
  If the algorithm rejects,
  because of Item~\eqref{item:hardness-second-order-reject-to-type}
  of \cref{claim:hardness-second-order},
  for every superset of \(C_1\),
  there is a superset of \(C_2\) of the same \(q\)-type.
  Hence, \(C_2\) is not needed anymore in the following computation,
  and we remove it from \(V_1\).
  Because of Item~\eqref{item:hardness-second-order-type-to-reject}
  of \cref{claim:hardness-second-order},
  after running the oracle for every pair in \(V_1\),
  there is at most one set in \(V_1\) for every \((2q+6)\)-type of sets of vertices in \(G\).
  Thus, the number of remaining sets in \(V_1\) can be bounded in terms of \(\abs{\phi}\)
  without any dependence on \(\abs{G}\).

  Next, in round \(i+1\) for \(i > 0\),
  we set \(V_{i+1} = \bigsetc{C \cup \set{v}}{C \in V_i,\ v \in V(G)}\).
  As in the first round, we run the oracle for every pair \(C_1, C_2 \in V_{i+1}\),
  and we remove \(C_2\) from \(V_{i+1}\) if the oracle rejects.
After the \(n\)th round, we set \(W = \bigcup_{i=1}^n V_i\).
  Again, we call the oracle for every pair of distinct sets in \(W\),
  and we remove one of the two sets if the oracle rejects.
  At the end, the number of sets in \(W\) only depends on \(\abs{\phi}\).
  Furthermore, for every set \(V' \subseteq V(G)\),
  there is a candidate \(C \in W\) of the same \(q\)-type.
  The number of oracle calls is quadratic in \(\abs{G}\),
  the size of the input graphs for the oracle calls is linear in the size of \(G\),
  and the parameters of the oracle calls can be bounded in terms of \(\abs{\phi}\).

  Recall that we assumed the input formula \(\phi\) to be of the form
  \(\exists X \psi\).
  For every candidate \(C \in W\), we recursively want to check whether \(G \models \psi(C)\) holds.
  For that, we extend the set of labels \(\Lambda\) by a new label \(X\),
  and we let \(G_C\) be the extension of \(G\) to the new label set with \(X(G_C) = C\).
  We can then consider \(\psi\) as a sentence over the new signature and recursively run our algorithm
  to check whether \(G_C \models \psi\) holds.

  Overall, considering both first-order and second-order quantification,
  we have a recursion tree where both the depth and the degree of the tree
  can be bounded in terms of \(\abs{\phi}\).
  In every single recursive call, the number of oracle queries is at most quadratic in \(\abs{V(G)}\),
  the size of the input graphs for the oracle queries is linear in the size of \(G\),
  the parameter \(\abs{\Lambda} + q + \ell\) for the oracle queries
  can be bounded in terms of \(\abs{\phi}\),
  and every oracle query only uses training sequences of length at most \(2\).
  Thus, the described procedure yields an fpt Turing reduction
  from \(\MSOMc\) to \(\uMSOLearn\).
\end{proof}
 \section{PAC Learning in Higher Dimensions}
\label{sec:pac}

So far, we considered the consistent-learning setting,
where the goal is to return a hypothesis that is consistent with the given examples.
In this section, we study the \MSO-learning problem in the agnostic PAC-learning setting.
There, for an instance space \(\X\),
we assume an (unknown) probability distribution \(\D\) on \(\X \times \set{+,-}\).
The learner's goal is to find a hypothesis \(h \colon \X \to \set{+,-}\),
using an oracle to draw training examples randomly from \(\D\),
such that \(h\) (approximately) minimizes the generalization error
\[\err_\D(h) \deff \Pr_{(x,\lambda) \sim \D} \bigl(h(x) \neq \lambda\bigr).\]
For every \(\Lambda\)-labeled graph \(G\) and \(q, k, \ell \in \NN\),
let \(\Hypo_{q, k, \ell}(G)\) be the hypothesis class
\[\Hypo_{q,k,\ell}(G) \deff \bigsetc{h_{\phi, \vec{w}}}{\phi \in \MSO(\Lambda, q, k, \ell, 0), \vec{w} \in (V(G))^\ell}.\]
Formally, we define the \MSO PAC-learning problem as follows.

\problemDefPara{\PacMSOLearn}
{\(\Lambda\)-labeled graph \(G\),
  numbers \(k, \ell, q \in \NN\), \(\delta, \epsilon \in (0,1)\),
  oracle access to probability distribution \(\D\) on \((V(G))^k \times \set{+,-}\)}
{\(\kappa \deff \abs{\Lambda} + k + \ell + q + 1/\delta + 1/\epsilon\)}
{Return a hypothesis \(h_{\phi, \vec{w}} \in \Hypo_{q, k, \ell}(G)\) such that,
with probability of at least \(1 - \delta\) over the choice of examples drawn
i.i.d.\ from \(\D\), it holds that
\[\err_\D(h_{\phi, \vec{w}})
  \leq \min_{h \in \Hypo_{q, k, \ell}(G)} \err_\D (h) + \epsilon.\]}

The remainder of this section is dedicated to the proof of \cref{Theorem:tractability-pac},
that is, we want to show that \PacMSOLearn is fixed-parameter linear
on classes of bounded clique-width
when the input graph is given as a \(c\)-expression.
To solve the problem algorithmically,
we can follow the \emph{Empirical Risk Minimization (ERM)} rule~\cite{vapnik_erm_1991,shalev-shwartz_understanding_2014},
that is, our algorithm should minimize the \emph{training error}
(or \emph{empirical risk})
\[\err_\TS(h) \deff \frac{1}{\abs{\TS}} \cdot \abs{\setc{(\vec{v}, \lambda) \in \TS}{h(\vec{v}) \neq \lambda}}\]
on the training sequence \(\TS\) of queried examples.
Roughly speaking, an algorithm can solve \PacMSOLearn
by querying a sufficient number of examples and then following the ERM rule.
To bound the number of needed examples,
we combine a fundamental result of statistical learning~\cite{blumer_learnability_1989,shalev-shwartz_understanding_2014},
which bounds the number of needed examples
in terms of the VC dimension of a hypothesis class,
with \cref{lem:gt}, a result due to Grohe and Tur{\'{a}}n~\cite{grohe_learnability_2004},
which bounds the VC dimension of \MSO-definable hypothesis classes
on graphs of bounded clique-width.

\begin{lemma}[\cite{blumer_learnability_1989,shalev-shwartz_understanding_2014}]
  A hypothesis class \(\Hypo\) is (agnostically) PAC-learnable if and only if it has finite VC dimension.
  Furthermore, if \(\Hypo\) has finite VC dimension, then \(\Hypo\) can be learned by any algorithm
  that follows the Empirical Risk Minimization rule with sample complexity
  (\ie, the number of queried examples needed to fulfil the bounds required for agnostic PAC learnability)
  \(m_\Hypo(\epsilon, \delta) \in \bigO\left(\frac{\VC(\Hypo) + \log(1/\delta)}{\epsilon^2}\right)\).
\end{lemma}

Combined with \cref{lem:gt}, this implies the following.

\begin{lemma}
  \label{lemma:sample-complexity}
  There is a computable function \(m \colon \NN^5 \times (0,1)^2 \to \NN\)
  such that any algorithm that proceeds as follows solves the problem \(\PacMSOLearn\).
  Given a \(\Lambda\)-labeled graph \(G\) of clique-width at most \(c\),
  numbers \(k, \ell, q \in \NN\), \(\delta, \epsilon \in (0,1)\),
  and oracle access to a probability distribution \(\D\) on \((V(G))^k \times \set{+,-}\),
  the algorithm queries at least \(m(c, \abs{\Lambda}, q, k, \ell, \delta, \epsilon)\)
  many examples from \(\D\) and then follows the ERM rule.
\end{lemma}

Using this lemma, we can now prove \cref{Theorem:tractability-pac},
showing that \(\PacMSOLearn\)
is fixed-parameter linear on classes of bounded clique-width
if the input graph is given as a \(c\)-expression,
even for dimensions \(k > 1\).

\begin{proof}[Proof of \cref{Theorem:tractability-pac}]
  Let \(G\) be a \(\Lambda\)-labeled graph,
  let \(k, \ell, q \in \NN\), \(\delta, \epsilon \in (0,1)\),
  and assume we are given oracle access to a probability distribution
  \(\D\) on \((V(G))^k \times \set{+,-}\).
  Moreover, let \(c \in \NN\) and let \(\Xi\) be a \(c\)-expression
  given as input that describes \(G\).

  Let \(s \deff m(c, \abs{\Lambda}, q, k, \ell, \delta, \epsilon)\),
  where \(m\) is the function from \cref{lemma:sample-complexity}.
  We sample \(s\) examples from \(\D\) and call the resulting
  sequence of training examples \(\TS\).
  Then, for every subsequence \(\TS'\) of \(\TS\),
  we will make use of the techniques in \cref{sec:tract}
  (adapted to higher-dimensional training data)
  and compute a hypothesis \(h_{\TS'} \in \Hypo_{q, k, \ell}\) that is consistent with \(\TS'\)
  if such a hypothesis exists.
  Finally, from all subsequences with a consistent hypothesis,
  we choose a subsequence \(\TS^*\) of maximum length and return \(h_{\TS^*}\).
  Note that this procedure minimizes the training error on the training sequence \(\TS\),
  \ie, \(\err_\TS(h_{\TS^*}) = \min_{h \in \Hypo_{q, k, \ell}} \err_\TS(h)\).
  Hence, the procedure follows the ERM rule, and, by \cref{lemma:sample-complexity},
  it solves \(\PacMSOLearn\).

  To apply the techniques from \cref{sec:tract},
  we first need to encode the training examples into the graph.
  For that, we extend the label set \(\Lambda\) by \(k \cdot s\) new labels
  \(\TS_{i,j}\) for \(i \in [s]\) and \(j \in [k]\),
  and we call the resulting label set \(\Lambda'\).
  Let \[\TS' = \Bigl(\bigl((v_{1,1}, \dots, v_{1,k}), \lambda_1\bigr),
    \dots, \bigl((v_{s',1}, \dots, v_{s',k}), \lambda_{s'}\bigr)\Bigr)\]
  be a subsequence of \(\TS\) for which we want to find a consistent hypothesis.
  Furthermore, let \(G_{\TS'}\) be the \(\Lambda'\)-labeled graph with vertex set \(V(G_{\TS'}) = V(G)\),
  \(R(G_{\TS'}) = R(G)\) for all \(R \in \Lambda\),
  \(S_{i,j}(G_{\TS'}) = \set{v_{i,j}}\) for all \(i \in [s']\) and \(j \in [k]\),
  and \(S_{i,j}(G_{\TS'}) = \emptyset\) for all \(i \in [s'+1,s]\) and \(j \in [k]\).

  Then, for every formula \(\phi(\vec{x}, \vec{y}) \in \MSO(\Lambda, q, k, \ell, 0)\),
  we consider the \(\MSO(\Lambda', q + k, 0, \ell, 0)\) formula
  \begin{align*}
    \phi'(y_1, \dots, y_\ell) & = \forall x_1 \dots \forall x_k\\
    \Biggl(
    &\quad \biggl(\Lor_{\substack{i \in [s']\\\lambda_i = +}} \Land_{j=1}^k S_{i,j}(x_j)\biggr) \rightarrow \phi(x_1, \dots, x_k, y_1, \dots, y_\ell)\\
    &\land \biggl(\Lor_{\substack{i \in [s']\\\lambda_i = -}} \Land_{j=1}^k S_{i,j}(x_j)\biggr) \rightarrow \neg \phi(x_1, \dots, x_k, y_1, \dots, y_\ell)
    \Biggr).
  \end{align*}
  Now, for every \(w_1, \dots, w_\ell \in V(G)\),
  we have \(G_{\TS'} \models \phi'(w_1, \dots, w_\ell)\)
  if and only if \(h_{\phi, \vec{w}}\) with \(\vec{w} = (w_1, \dots, w_\ell)\)
  is consistent with \(\TS'\).
  Hence, analogously to the procedure described in \cref{sec:tract},
  we can find a consistent hypothesis, if it exists,
  by going through all possible formulas \(\phi\)
  and computing \(w_1, \dots, w_\ell \in V(G)\) with \(G_{\TS'} \models \phi'(w_1, \dots, w_\ell)\).

  Given a \(c\)-expression for \(G\),
  we can compute a \((c + k \cdot s)\)-expression for \(G_{\TS'}\) in linear time.
  Moreover, \(G_{\TS'}\) and \(G\) have the same vertex set,
  and \(\abs{\Lambda'}\) only depends on \(\abs{\Lambda}\), \(k\),
  and \(s = m(c, \abs{\Lambda}, q, k, \ell, \delta, \epsilon)\).
  Thus, analogously to \cref{sec:tract},
  a consistent hypothesis for a single subsequence \(\TS'\)
  can be computed by an fpt-algorithm with parameters
  \(c, \abs{\Lambda}, k, \ell, q, \delta\), and \(\epsilon\).
  The number of subsequences to check can be bounded by \(2^s\).
  Hence, all in all, the described procedure is an fpt-algorithm
  that solves \(\PacMSOLearn\).
\end{proof}
 \section{Consistent Learning in Higher Dimensions}
\label{sec:high-dim-cons}

In this section, we consider the problem \MSOLearn with dimension $k > 1$ on labeled graphs of bounded clique-width.
A brute-force attempt yields a solution in time $g(c,\abs{\Lambda},q,k,\ell) \cdot (V(G))^{\ell+1} \cdot m$,
where \(m\) is the length of the training sequence,
for some $g \colon \NN^5 \to \NN$.
This is achieved by iterating over all formulas, then iterating over all parameter assignments,
and then performing model checking for each training example.
We assume that the graph $G$ is considerably larger in scale than the sequence of training examples $S$.
Therefore, \cref{thm:tract-high-dim} significantly improves the running time to \(\bigO \bigl((m+1)^{g(c, \abs{\Lambda}, q, k, \ell)}|V(G)|^2\bigr)\).
While \cref{thm:tract-high-dim} is not a fixed-parameter tractability result in the classical sense,
we show that this is optimal in \cref{sec:high-dim-hardness}.
The present section is dedicated to the proof of \cref{thm:tract-high-dim}.

Until now, we have viewed the training sequence as a sequence of tuples
\(S \in ((V(G))^k \times \{+,-\})^m\). In the following, it is useful
to split the training sequence into two parts, a sequence of
vertex tuples $\Ca \in ((V(G))^k)^m$ and a function $\sigma \colon [m]\to\{+,-\}$ which
assigns the corresponding label to each tuple.
Let $G$ be a $\Lambda$-labeled graph,
$\Ca=(\vec v_1,\ldots,\vec v_m)\in ((V(G))^k)^m$,
$\sigma \colon [m]\to\{+,-\}$,
and $\phi(\vec x,\vec y)\in\MSO(\Lambda,q,k,\ell,0)$.
We call $(G,\Ca,\sigma)$ \emph{$\phi$-consistent} if there is a parameter setting
$\vec w\in (V(G))^\ell$ such that for all $i\in[m]$,
\(G\models\phi(\vec v_i,\vec w)\iff \sigma(\vec v_i)=+\).
We say that $\vec w$ is a \emph{$\phi$-witness} for $(G,\Ca,\sigma)$.
This notation allows us to state the main technical ingredient
for the proof of \cref{thm:tract-high-dim} as follows.

\begin{restatable}{theorem}{theoComputeConsist}
\label{theo:compute-consist}
  There is a computable function $g \colon \Nat^5\to\Nat$ and
  an algorithm that, given a $\Lambda$-graph $G$ of clique-width $\cw(G)\leq c$, a
  sequence $\Ca=(\vec v_1,\ldots,\vec v_m)\in ((V(G))^k)^m$, a function
  $\sigma \colon [m]\to\{+,-\}$, and a formula $\phi(\vec x,\vec
  y)\in\MSO(\Lambda,q,k,\ell,0)$, decides if $(G,\Ca,\sigma)$ is
  $\phi$-consistent in time
  \(
    (m+1)^{g(c,|\Lambda|,q,k,\ell)}|G|
  \).
\end{restatable}

Using \cref{theo:compute-consist}, we can now prove \cref{thm:tract-high-dim}.

\begin{proof}[Proof of \cref{thm:tract-high-dim}]
Given a \(\Lambda\)-labeled graph \(G\) and a training sequence
\(S = ((\vec v_1, \lambda_1), \dots,\) \((\vec v_m, \lambda_m))
\in ((V(G))^k \times \{+,-\})^m\),
we let
$\Ca \deff (\vec v_1,\ldots,\vec v_m) \in ((V(G))^k)^m$ and
$\sigma \colon [m]\to\{+,-\}, i \mapsto \lambda_i$.
We iterate over all formulas
$\phi \in \MSO(\Lambda,q,k,\ell,0)$ and use \cref{theo:compute-consist}
to check whether $(G,\Ca, \sigma)$ is $\phi$-consistent.
If there is no \(\phi\) such that \((G, \Ca, \sigma)\) is \(\phi\)-consistent, then we reject the input.
Otherwise, let \(\phi \in \MSO(\Lambda, q, k, \ell, 0)\) be such that
\((G, \Ca, \sigma)\) is \(\phi\)-consistent.
We compute a $\phi$-witness following the same construction as in the proof of \cref{lem:unary-tractability-general}.
That is, using a fresh label, we encode the parameter choice of a single variable into the graph,
and then we check whether consistency still holds for the corresponding formula $\phi'$ that enforces this parameter choice.
In total, we perform up to \(\ell \cdot \abs{V(G)}\) such consistency checks to compute a $\phi$-witness
$\vec w$.
The consistent formula $\phi$ together with the $\phi$-witness
$\vec w$ can then be returned as a hypothesis $h_{\phi, \vec w}$ that
is consistent with $S$ on $G$ and therefore a solution to $\MSOLearn$.
\end{proof}

The remainder of this section is dedicated to the proof of \cref{theo:compute-consist}.
Intuitively, we proceed as follows.

Let \(\Xi\) be a \(c\)-expression describing the input graph \(G\).
As described in \cref{sub:prelim:clique-width}, we may view \(\Xi\) as a tree.
In a bottom-up algorithm, starting at the leaves of \(\Xi\),
we compute the set of all tuples \(\btheta = (\theta_1, \dots, \theta_m)\) of \((\Lambda, q, k_i, \ell)\)-types \(\theta_i\)
that are realizable over \(\Ca\) in \(G\).
Hence, to check whether \((G, \Ca, \sigma)\) is \(\phi\)-consistent,
it suffices to check whether there is a realizable type \(\btheta\)
such that \(\phi \in \theta_i\) if and only if \(\sigma(i) = +\) for all \(i \in [m]\).
We explain this in more detail in \cref{sub:tract-highdim-algorithm}.
The difficulty with this approach is that we are talking about \(m\)-tuples of types.
In general, the number of such tuples is exponential in $m$,
and hence the size of the set we aim to compute could be exponentially large.
Fortunately, this does not happen in graphs of bounded clique-width.
By \cref{lem:gt}, we can bound the VC dimension of a first-order formula
over classes of graphs of bounded clique-width.
Further, we show in \cref{sub:tract-highdim-vc}
that this suffices to give a polynomial bound for the number of realizable tuples.
Then, in \cref{sub:tract-highdim-composition}, we show how to compute
the set of realizable types in a node of \(\Xi\)
based on the realizable types in the children of the node.

\subsection{Bounding the Number of Realizable Types}
\label{sub:tract-highdim-vc}

In this subsection, we prove the following bound on the number of realizable tuples of types.

\begin{restatable}{lemma}{lemTypesBound}
\label{lem:types-bound}
  Let $d,q,k,\ell\in\Nat$, let $t \deff \abs{\Tp(\Lambda,q,k,\ell)}$,
  and let $G$ be a $\Lambda$-labeled graph such that
  $\VC(\phi,G)\le d$ for all $\phi\in\MSO(\Lambda,q,k,\ell,0)$. Let
  $\Ca\in (V(G))^{\vec k}$ for some $\vec k\in\{0,\ldots,k\}^m$. Then at most $(k{+}1) \cdot g(d,m)^t$ types
  in $\Tp(\Lambda,q,\vec k,\ell)$ are realizable over $\Ca$ in
  $G$.
\end{restatable}

The proof of \cref{lem:types-bound} is based on \cref{lem:sauer-shelah}.

\begin{lemma}[Sauer--Shelah Lemma \cite{sauer_density_72,shelah_density_72}]
  \label{lem:sauer-shelah}
  Let \(q \in \NN\),
  $\phi \in \MSO(\Lambda, q, k, \ell)$,
  let $\CC$ be a class of $\Lambda$-labeled graphs,
  and let \(d \in \NN\) such that $\VC(\phi,\CC) \leq d$. Then
  \[
    \abs{H_\phi(G,X)} \leq \sum_{i=0}^d\binom{\abs{X}}{i} \in \bigO\bigl(\abs{X}^d\bigr)
  \]
  for all $G\in\CC$ and all (finite) $X\subseteq (V(G))^k$.
\end{lemma}

Further, we need a simple combinatorial
argument, which is captured by the following \cref{lem:bonnet}. We remark that this
argument is known (see \cite[Proof of Lemma 23]{bonnet_tw_22}). For a
matrix $M\in\Sigma^{m\times n}$ and a symbol $\sigma\in\Sigma$, we let
$M^{(\sigma)}\in\{0,1\}^{m\times n}$ be the matrix obtained from $M$
by replacing all $\sigma$-entries by $1$ and all other entries by $0$.

\begin{restatable}[\cite{bonnet_tw_22}]{lemma}{lemMatrix}\label{lem:bonnet}
  Let $c,m,n,s\in\NN_{>0}$ such that $n>(c-1)^{s-1}$. Furthermore, let
  $\Sigma$ be a finite alphabet of size $|\Sigma|=s$, and let
  $M\in\Sigma^{m\times n}$ be a matrix with mutually distinct
  columns. Then there is a $\sigma\in\Sigma$ such that $M^{(\sigma)}$
  has at least $c$ distinct columns.
\end{restatable}
\begin{proof}
  Without loss of generality, we assume that $\Sigma=[s]$.
  Then, for every column $\vec v$ of $M$, we can write $\vec v=\sum_{\sigma=1}^s\sigma
  \vec v^{(\sigma)}$, where $\vec v^{(\sigma)}$ is the column of $M^{(\sigma)}$
  corresponding to $\vec v$.

  Suppose for contradiction that for every $\sigma$, the set
  $\CC^{(\sigma)}$ of distinct column vectors appearing in
  $M^{(\sigma)}$ has cardinality at most $c-1$. Then there are at most
  $(c-1)^{s-1}$ vectors $\sum_{\sigma=1}^s\sigma
  \vec v^{(\sigma)}$ with $\vec v^{(\sigma)}\in\CC^{(\sigma)}$ that have $1$s in
  mutually disjoint positions, because once we
  have chosen $\vec v^{(1)},\ldots,\vec v^{(s-1)}$,
  there is only one choice left for $\vec v^{(s)}$.
  This contradicts our assumption that all columns of $M$ are distinct.
\end{proof}

In the following, for a set \(K \subseteq \NN\), we let
\[
  \Tp(\Lambda, q, K^m, \ell) \deff \bigcup_{\vec{k} \in K^m} \Tp(\Lambda, q, \vec{k}, \ell).
\]

\begin{proof}[Proof of \cref{lem:types-bound}]
  For $0\le p\le k$, let $\Ca^{(p)}$ be the subsequence of $\Ca$ consisting
  of all entries of length $p$, and let $m_p \deff \bigl\lvert\Ca^{(p)}\bigr\rvert$.
  Every type $\btheta\in\Tp(\Lambda,q,\vec k,\ell)$ projects to a
  type $\btheta^{(p)}\in\Tp(\Lambda,q,\{p\}^{m_p},\ell)$.

  Suppose for contradiction that more than $b \deff (k{+}1) \cdot g(d,m)^t$ types
  in $\Tp(\Lambda,q,\vec k, \ell)$  are
  realizable over $\Ca$ in $G$.

  If $\btheta$ is
  realizable over $\Ca$, then for each $p$, the projection $\btheta^{(p)}$ is
  realizable over $\Ca^{(p)}$.
  Moreover, if types $\btheta,\btheta'$ are distinct, then there is at
  least one $p$ such that $\btheta_p$ and $\btheta_p'$ are
  distinct. Thus, there is a $p$ such that at least
  $g(d,m)^t+1$ types in $\Tp(\Lambda,q,\{p\}^{m_p},\ell)$ are realizable over $\Ca^{(p)}$.

  Suppose $\Ca^{(p)}=(\vec v_1,\ldots,\vec v_{m_p})$,
  where $\vec{v}_i\in (V(G))^p$.
  Let $\btheta_1,\ldots,\btheta_n\in\Tp(\Lambda,q,\{p\}^{m_p},\ell)$
  be distinct types
  and $\vec w_1,\ldots,\vec w_n\in (V(G))^\ell$ for some \(n \geq g(d,m)^t + 1\)
  such that $\tp_q^G(\Ca^{(p)},\vec w_j)=\btheta_j$,
  that is, the types \(\btheta_1, \dots, \btheta_r\) are realizable over \(\Ca^{(p)}\).
  Suppose that
  $\btheta_j=(\theta_{1j},\ldots,\theta_{m_pj})$ with
  $\theta_{ij}\in\Tp(\Lambda,q,p,\ell)$. Then $\tp_q^G(\vec v_i,\vec w_j)=\theta_{ij}$. Let
  $\Sigma\coloneqq\Tp(\Lambda,q,p,\ell)$ and
  let $M\in\Sigma^{m_p\times n}$ be the matrix with entries
  $M_{ij}=\theta_{ij}$. Note that $|\Sigma|\le t$ and that the columns of $M$ are the vectors
  $\btheta_j$, which are mutually distinct. Since
  $n > g(d,m)^t\ge(g(d,m)+1-1)^{t-1}$, by \cref{lem:bonnet}, there is
  a $\theta\in\Sigma$ such that the matrix $M^{(\theta)}$ has at least
  $g(d,m)+1$ distinct columns.

  Observe that for
  \(i\in[m_p]\) and \(j\in[n]\), if $M^{(\theta)}_{ij}=1$, then $\tp_q^G(\vec
  v_i,\vec w_j)=\theta$, and if $M^{(\theta)}_{ij}=0$, then $\tp_q^G(\vec
  v_i,\vec w_j)\neq\theta$. Let
  \[
    \phi(\vec x,\vec y) \deff \bigwedge_{\psi(\vec x,\vec
      y)\in\theta}\psi(\vec x,\vec y).
  \]
  Then for $\vec v\in (V(G))^p$, $\vec w\in (V(G))^\ell$,
  \[
    G\models\phi(\vec v,\vec w)\iff\tp_q^G(\vec v,\vec w)=\theta.
  \]
  Let $X\coloneqq\{\vec v_1,\ldots, \vec v_{m_p}\}$. Then for all $j\in[n]$,
  \[
    X\cap\phi(G,\vec w_j)=\{\vec v_i\mid M^{(\theta)}_{ij}=1\}.
  \]
  Since the matrix $M^{(\theta)}$ has more than $g(d,m)\ge g(d,m_p)$ distinct
  columns, it follows that $|H_\phi(G,X)|>g(d,m_p)$,
  which contradicts \cref{lem:sauer-shelah}.
\end{proof}

\subsection{Compositionality}
\label{sub:tract-highdim-composition}
In the following, we show how the realizable tuples of types for an expression can be computed based on the realizable tuples of types of its subexpressions.

For the operators \(\eta_{P,Q}\) and \(\rho_{P,Q}\),
this is done by adding edges and relabeling.

\begin{lemma}\label{lem:eta-rho}
  Let $\Lambda$ be a label set, \(q, k \in \NN\), $P,Q\in\Lambda$, and \(\xi \in \set{\eta, \rho}\).
  For every type  $\theta\in\Tp(\Lambda,q,k)$, there is a set
  $T_{\xi,P,Q}(\theta)\subseteq\Tp(\Lambda,q,k)$ such that
  for every $\Lambda$-graph \(G'\), for $G \deff \xi_{P,Q}(G')$,
  and for every $\vec v\in (V(G))^k$,
  \[
    \tp_q^G(\vec v)=\theta\iff\tp_q^{G'}(\vec
    v)\in T_{\xi,P,Q}(\theta).
  \]
  Furthermore, the mappings
  \(T_{\eta,P,Q}\) and \(T_{\rho, P, Q}\) are computable.
\end{lemma}

\begin{proof}
  This follows from the Theorem on Syntactic Interpretations \cite[Chapter VIII]{ebbinghaus_mathematicallogic_2021}.
\end{proof}

Note that, for \(\xi \in \set{\eta, \rho}\) and for distinct $\theta_1,\theta_2$, the sets
$T_{\xi,P,Q}(\theta_1)$ and $T_{\xi,P,Q}(\theta_2)$ are mutually
disjoint up to types that are not realizable. Thus, for every
realizable $\theta'\in \Tp(\Lambda,q,k)$, there is at
most one $\theta\in \Tp(\Lambda,q,k)$ such that $\theta'\in T_{\xi,P,Q}(\theta)$.
\begin{corollary}\label{cor:eta-rho}
  Let $\Lambda$ be a label set, \(q, \ell \in \NN\), $P,Q\in\Lambda$, \(\xi \in \set{\eta, \rho}\),
  $\vec k\in\Nat^m$,
  and $\btheta=(\theta_1,\ldots,\theta_m)\in\Tp(\Lambda,q,\vec k,\ell)$.
  Let \(G'\) be a \(\Lambda\)-graph, let \(G \deff \xi_{P,Q}(G')\),
  and $\Ca\in (V(G))^{\vec k}$.
  Then $\btheta$ is realizable over $\Ca$ in $G$ if and only there is a
  $\btheta'=(\theta_1',\ldots,\theta_m')\in\Tp(\Lambda,q,\vec{k},\ell)$
  such that $\btheta'$ is realizable over $\Ca$ in $G'$
  and $\theta_i'\in T_{\xi,P,Q}(\theta_i)$ for all $i\in[m]$,
  where $T_{\xi,P,Q}$ is the mapping of Lemma~\ref{lem:eta-rho}.
\end{corollary}

Next, we handle the operator \(\delta_P\) that deletes the label \(P\).

\begin{lemma}\label{lem:delta}
  Let $\Lambda$ be a label set, let \(P\) be a label with $P\not\in\Lambda$,
  and let \(q, k \in \NN\).
  For every type  $\theta\in\Tp(\Lambda,q,k)$, there is a set
  $T_{\delta,P}(\theta) \subseteq \Tp(\Lambda\cup\{P\},q,k)$ such that
  for every $(\Lambda \cup \set{P})$-graph \(G'\), for $G \deff \delta_{P}(G')$,
  and for every $\vec v\in (V(G))^k$,
  \[
    \tp_q^G(\vec v)=\theta\iff\tp_q^{G'}(\vec
    a)\in T_{\delta,P}(\theta).
  \]
  Furthermore, the mapping
  $T_{\delta,P}$ is
  computable.
\end{lemma}

\begin{proof}
  This follows from the Theorem on Syntactic Interpretations \cite{ebbinghaus_mathematicallogic_2021}.
\end{proof}

\begin{corollary}\label{cor:delta}
  Let $\Lambda$ be a label set, let \(P\) be a label with $P\not\in\Lambda$,
  let \(q, \ell \in \NN\),
  $\vec k\in\Nat^m$, and $\btheta=(\theta_1,\ldots,\theta_m)\in\Tp(\Lambda,q,\vec k,\ell)$.
  Let \(G'\) be a $(\Lambda \cup \set{P})$-graph, let $G \deff \delta_{P}(G')$,
  and $\Ca\in (V(G))^{\vec k}$.
  Then $\btheta$ is realizable over $\Ca$ in $G$ if and only there is a
  $\btheta'=(\theta_1',\ldots,\theta_m')\in\Tp(\Lambda\cup\{P\},q,\vec
  k,\ell)$ such that $\btheta'$ is realizable over $\Ca$ in $G'$ and $\theta_i'\in
  T_{\delta,P}(\theta_i)$ for all $i\in[m]$, where $T_{\delta,P}$ is
  the mapping of Lemma~\ref{lem:delta}.
\end{corollary}

Finally, we handle the ordered-disjoint-union operator \(\uplus^<\).
For a $k$-tuple $\vec v=(a_1,\ldots,a_k)$ and a set $I\subseteq [k]$,
say, $I=\{i_1,\ldots,i_p\}$ with $i_1<i_2<\ldots<i_p$, we let
$\vec v_I\coloneqq(a_{i_1},\ldots,a_{i_p})$.

\begin{lemma}\label{lem:du}
  Let $\Lambda$ be a label set with $P_1^<,P_2^<\in\Lambda$.
  Let $q, k\in\Nat$ and $K_1\subseteq[k],K_2=[k]\setminus K_1$,
  and let $k_j\coloneqq |K_j|$ for $j=1,2$.

  For every type $\theta\in\Tp(\Lambda,q,k)$ and for $j=1,2$,
  there is a type
  $T_{\uplus,j,K_j}(\theta)\in\Tp(\Lambda\setminus\{P_1^<,P_2^<\},q,k_j)$,
  such that
  for every $\Lambda$-graph $G=G_1\uplus^<G_2$ and every $\vec
  v\in (V(G))^k$ with \(\vec v_{K_1}\in (V(G_1))^{k_1}\)
  and \(\vec v_{K_2}\in (V(G_2))^{k_2}\), we have
  \begin{align*}
    \tp_q^G(\vec v)=\theta \; \iff \ \;
    &\tp_q^{G_1}(\vec{v}_{K_1}) = T_{\uplus,1,K_1}(\theta)\text{ and }\\
    &\tp_q^{G_2}(\vec{v}_{K_2})= T_{\uplus,2,K_2}(\theta).
  \end{align*}
  Furthermore, for $j=1,2$, the mapping
  $T_{\uplus,j,K_j}$ is
  computable.
\end{lemma}

\begin{proof}
  This is a version of the Feferman--Vaught Theorem for \MSO \cite{makowsky_fefermanvaught_2004} that can easily be shown using Ehrenfeucht--Fra\"{i}ss\'{e} games.
\end{proof}

In the case of disjoint union, it is slightly more complicated to
compute the realizable types in $G=G_1\uplus^<G_2$ from the realizable
types in $G_1$ and $G_2$. We need some additional notation.
Let $V$ be a set and $W\subseteq V$. For a tuple $\vec
v=(v_1,\ldots,v_k)\in V^k$, we let $\vec v\cap W\coloneqq\vec v_I$ for
the set $I=\{i\in[k]\mid v_i\in W\}$. For a sequence $\Ca=(\vec
v_1,\ldots,\vec v_m)\in V^{\vec k}$ of tuples, we let
\[
  \Ca\cap W\coloneqq \big(\vec v_1\cap W,\ldots,\vec v_m\cap W\big).
\]
Note that even if all tuples in $\Ca$ have the same length, this is
not necessarily the case for $\Ca\cap W$; some tuples in
$\Ca\cap W$ may even be empty.

In the following corollary, we consider types in $\Tp(\Lambda,q,k,\ell)$
instead of types in $\Tp(\Lambda,q,k)$. Thus, we need to parameterize the mappings $T_{\uplus,j,\ldots}$
by pairs of sets $K_j\subseteq[k],L_j\subseteq[\ell]$: for
$j=1,2$, we obtain a mapping
$T_{\uplus,j,K_j,L_j} \colon \Tp(\Lambda,q,k)\to \Tp(\Lambda\setminus\{P_1^<,P_2^<\}, q, |K_j|, |L_j|)$.

\begin{corollary}\label{cor:du}
   Let $\Lambda$ be a label set with $P_1^<,P_2^<\in\Lambda$,
   let $q, \ell\in\Nat$,
   $\vec k=(k_1,\ldots,k_m)\in\Nat^m$, and let
   $\btheta=(\theta_1,\ldots,\theta_m)\in\Tp(\Lambda,q,\vec
   k,\ell)$. Let $G=G_1\uplus^<G_2$ be a \(\Lambda\)-graph and $\Ca=(\vec v_1,\ldots,\vec v_m)\in
   (V(G))^{\vec k}$. For every $j\in[2]$, $i\in[m]$, let
   $K_{ji}\subseteq[k_i]$ be such that $\vec v_i\cap V(G_j)=(\vec
   v_i)_{K_{ji}}$.

   Then $\btheta$ is realizable over $\Ca$ in $G$ if and
   only if there are $L_1,L_2\subseteq [\ell]$ such that
   $L_2=[\ell]\setminus L_1$ and for $j=1,2$,
   \[
     \btheta_j=\big(T_{\uplus,j,K_{j1},L_j}(\theta_{1}),\ldots,
     T_{\uplus,j,K_{jm},L_j}(\theta_{m})\big)
   \]
   is realizable over $\Ca\cap V(G_j)$ in $G_j$.
\end{corollary}

\subsection{Computing the Realizable Types}
\label{sub:tract-highdim-algorithm}

For the proof of \cref{theo:compute-consist}, we use the following result
that allows us to compute the realizable types of an expression.

\begin{lemma}
  \label{lem:realisable}
  There is a computable function $f \colon \Nat^4\to\Nat$ and an algorithm that, given
  $c,q,k,\ell,m\in\Nat$, a vector $\vec
  k=(k_1,\ldots,k_m)\in\Nat^m$ with $k_i\le k$ for all $i\in[m]$, a
  $\Lambda$-expression $\Xi$ with $|\Lambda|\le c$, and
  a sequence $\Ca\in (V_\Xi)^{\vec k}$,
  computes the set of all $\btheta \in \Tp(\Lambda,q,\vec k,\ell)$
  that are realizable over $\Ca$ in $G_{\Xi}$ in time
  \[
    \bigO\Big((m+1)^{f(c,q,k,\ell)}\cdot|\Xi|\Big).
  \]
\end{lemma}

\begin{proof}
  As argued in \cref{sec:prelim}, we may assume that $\Xi$ only contains
  ordered-disjoint-union operators and no plain disjoint-union
  operators.

  For every subexpression $\Xi'$, we let $\Lambda_{\Xi'}$ be the set of
  labels of $\Xi'$, that is, the set of unary relation
  symbols such that $G_{\Xi'}$ is a $\Lambda_{\Xi'}$-graph. Moreover,
  let $\Ca_{\Xi'}\coloneqq\Ca\cap V_{\Xi'}$, and let
  $\vec k_{\Xi'}\subseteq\Nat^m$ such that
  $\Ca_{\Xi'}\in (V_{\Xi'})^{\vec k_{\Xi'}}$.

  We inductively construct, for every subexpression $\Xi'$ of $\Xi$
  and $0\le\ell'\le\ell$, the set $ \CR_{\ell'}(\Xi')$ of all types
  $\btheta\in \Tp(\Lambda_{\Xi'},q,\vec k_{\Xi'},\ell')$ that are
  realizable over $\Ca_{\Xi'}$ in $G_{\Xi'}$.
\begin{description}
    \item[Case 1: $\Xi'$ is a base expression.]
      In this case, for each $\ell' \in [\ell]$, we can construct $\CR_{\ell'}(\Xi')$ by brute force in time
      $f_1(c,q,k,\ell) \cdot m$ for a suitable (computable) function $f_1$.
      Let \((k_1', \dots, k_m') \deff \vec{k}_{\Xi'}\).
      We compute $\theta_i$ by iterating over all formulas $\phi$ with $k_i' + \ell'$ free variables and evaluating $\phi$ on the single vertex graph $G_{\Xi'}$.

    \item[Case 2: $\Xi'=\eta_{P,Q}(\Xi'')$.]
      Let $0\le\ell'\le\ell$, and let $\vec k'=(k_1',\ldots,k_m')\coloneqq\vec k_{\Xi'}=\vec k_{\Xi''}$.
      As we show in \cref{lem:eta-rho,cor:eta-rho},
      there is a computable mapping \(T_{\eta,P,Q} \colon \Tp(\Lambda_{\Xi'}) \to 2^{\Tp(\Lambda_{\Xi''})}\) such that
      $\CR_{\ell'}(\Xi')$ is the set of all $\btheta=(\theta_1,\ldots,\theta_m)\in
      \Tp(\Lambda_{\Xi'}, q, \vec k',\ell')$ such that there is a
      $\btheta'=(\theta_1',\ldots,\theta_m')\in\CR(\Xi'')$ with $\theta_i'\in
      T_{\eta,P,Q}(\theta_i)$ for all $i\in[m]$.
      Moreover, for every realizable $\theta'\in \Tp(\Lambda_{\Xi'})$,
      we guarantee that there is at most one type $\theta\in \Tp(\Lambda_{\Xi''})$ such that $\theta'\in T_{\eta,P,Q}(\theta)$.
      To compute the set $\CR_{\ell'}(\Xi')$, we
      step through all
      $\btheta'\in\CR(\Xi'')$. For each such $\btheta'=(\theta_1',\ldots,\theta_m')$,
      for all $i\in[m]$,
      we compute the unique $\theta_i\in\Tp(\Lambda_{\Xi'}, q, k_i', \ell')$ such that
      $\theta_i'\in T_{\eta,P,Q}(\theta_i)$. If for some $i\in[m]$, no such
      $\theta_i$ exists, we move on to the next $\btheta'$. Otherwise,
      we add $\btheta=(\theta_1,\ldots,\theta_m)$ to $\CR(\Xi')$.

    \item[Case 3: $\Xi'=\rho_{P,Q}(\Xi'')$.]
      Analogous to Case~2, again using \cref{lem:eta-rho,cor:eta-rho}.

    \item[Case 4: $\Xi'=\delta_{P}(\Xi'')$.]
      Analogous to Case~2, using \cref{lem:delta,cor:delta}.

    \item[Case 5: $\Xi'=\Xi_1\uplus^<\Xi_2$.]
      Let $\Lambda'\coloneqq\Lambda_{\Xi'}$, $V'\coloneqq V_{\Xi'}$, $\vec k'=(k_1',\ldots,k_m')\coloneqq\vec
      k_{\Xi'}$, and  $\Ca' \deff (\vec v_1',\ldots,\vec v_m') \deff \Ca_{\Xi'}=\Ca\cap V'$.
      For $j=1,2$, let $\Lambda_j\coloneqq\Lambda_{\Xi_j}$,
      $V_j\coloneqq V_{\Xi_j}$, $\vec
      k_j \deff (k_{j1},\ldots,k_{jm})\coloneqq\vec k_{\Xi_j}$,
      and for all $i\in[m]$, let
      $K_{ji}\subseteq[k_{ji}]$ such that $\vec v_i'\cap V_j=(\vec v_i)_{K_{ji}}$. Let $0\le
      \ell'\le\ell$.

      For
      all $L_1,L_2\subseteq[\ell']$ such that $L_2=[\ell']\setminus
      L_1$, we let $\CR_{L_1,L_2}$ be the set of all
      $\btheta=(\theta_1,\ldots,\theta_m)\in\Tp(\Lambda',q,\vec
      k',\ell')$ such that for $j=1,2$, we have
      \[
        \qquad\quad
        \btheta_j \deff \big(T_{\uplus,j,L_j,K_{j1}}(\theta_{1}),\ldots,
        T_{\uplus,j,L_j,K_{jm}}(\theta_{m})\big)\in\CR_{|L_j|}(\Xi_j),
      \]
      where \(T_{\uplus,j,L_j,K_{ji}} \colon \Tp(\Lambda') \to \Tp(\Lambda_j)\) for \(i \in [m]\) is a computable mapping
      that we give in \cref{lem:du}.
      Then, as we show in \cref{cor:du},
      \[
        \CR_{\ell'}(\Xi')=\bigcup_{\substack{L_1\subseteq[\ell']\\L_2=[\ell']\setminus
            L_1}}\CR_{L_1,L_2}.
      \]
      To compute $\CR_{L_1,L_2}$, we iterate over all
      $\btheta_1=(\theta_{11},\ldots,\theta_{1m})\in\CR_{|L_1|}(\Xi_1)$. For
      all $i\in[m]$ we compute the unique $\theta_i\in\Tp(\Lambda',q,k_i',\ell')$ such that
      $T_{\uplus,1,L_1,K_{1i}}(\theta_i)=\theta_{1i}$. If, for some $i\in[m]$, no such
      $\theta_i$ exists, then we move on to the next $\btheta_1$. Otherwise,
      we compute
      \[
        \btheta_2=\big(T_{\uplus,2,L_2,K_{21}}(\theta_{1}),\ldots,
        T_{\uplus,2,L_2,K_{2m}}(\theta_{m})\big)
      \]
      and check if $\btheta_2\in\CR_{|L_2|}(\Xi_2)$. If it is, we add
      $\btheta$ to $\CR_{L_1,L_2}$. Otherwise, we move on to the next $\btheta_1$.
  \end{description}
  This completes the description of our algorithm. To analyze the
  running time, let
  \[
    r\coloneqq\max_{\Xi'}|\CR_{\Xi'}|,
  \]
  where $\Xi'$ ranges over all subexpressions of $\Xi$. By
  \cref{lem:gt,lem:types-bound},
  there is a computable function
  $f_2 \colon \Nat^4\to\Nat$ such that
  \[
    r\le (m+1)^{f_2(c,q,k,\ell)}.
  \]
  The running time of each step of the constructions can be bounded by
  $f_3(c,q,k,\ell)\cdot r$ for a suitable computable function $f_3$. We need to
  make $|\Xi|$ steps. Thus,
  overall, we obtain the desired running time.
\end{proof}

Finally, we can finish the proof of \cref{theo:compute-consist}
by checking whether the realizable types match with the given formula \(\phi\).

\begin{proof}[Proof of \cref{theo:compute-consist}]
    We assume that the input graph \(G\) is given as a
    $c$-expression. To check whether
    \((G, \Ca, \sigma)\)
    is $\phi$-consistent, we compute the set $\CR$ of
    all $\btheta \in \Tp(\Lambda,q,\Bar{k},\ell)$
    that are realizable over \(\Ca\) in \(G\),
    using \cref{lem:realisable}.
    Then, for each \(\btheta = (\theta_1, \ldots, \theta_m ) \in \CR\)  we check if \(\phi \in \theta_i \Longleftrightarrow \sigma(i) = +\).
    If we find such a \(\btheta\), then \((G, \Ca, \sigma)\) is \(\phi\)-consistent; otherwise it is not.
\end{proof}
 \section{Hardness of Checking Consistency in Higher Dimensions}
\label{sec:high-dim-hardness}

The following result shows,
under the assumption \(\FPT \neq \Wone\),
that \cref{theo:compute-consist}
can not be improved to an fpt-result.

\begin{restatable}{theorem}{highDimHardness}
  \label{thm:high-dim-hardness}
  There is a \(q \in \NN\) such that the following parameterized problem
  is \(\Wone\)-hard.
\end{restatable}

  \problemDefParaNoName{graph \(G\) of clique-width at most \(2\),
    sequence \(\Ca = (\vec{a}_1, \dots, \vec{a}_m) \in \bigl((V(G))^2\bigr)^m\),
    function \(\sigma \colon [m] \to \set{+1, -1}\),
    formula \(\phi(\vec{x}, \vec{y}) \in \MSO(\Lambda, q, 2, \ell, 0)\)}
  {\(\ell\)}
  {decide if \((G, \Ca, \sigma)\) is \(\phi\)-consistent.}

\begin{proof}
  We prove this by a reduction from the \(\Wone\)-complete
  \emph{weighted satisfiability problem}
  for Boolean formulas in \(2\)-conjunctive normal form
  \cite{flum_parameterized_2006}.
  The \emph{weight} of an assignment to a set of Boolean variables
  is the number of variables set to \(1\).

  \problemDefPara{\(\problemFont{WSat(2-CNF)}\)}
  {Boolean formula \(\Phi\) in \(2\)-CNF}
  {\(\ell\)}
  {decide if \(\Phi\) has a satisfying assignment of weight \(\ell\).}

  Given a \(2\)-CNF formula \(\Phi = \Land_{i=1}^m (L_{i,1} \lor L_{i,2})\)
  in the variables \(\set{X_1, \dots, X_n}\) and \(\ell \in \NN\),
  we construct an instance of the consistency problem from \cref{thm:high-dim-hardness} as follows.

  We let \(G\) be the graph with vertex set
  \[V(G) \deff \setc{X_i, \neg X_i, Y_{i,1}, Y_{i,2}, Z_i}{i \in [n]}\]
  and edge set
  \[E(G) \deff \setc{(X_i, \neg X_i), (X_i, Y_{i,1}), (X_i, Y_{i,2}), (\neg X_i, Z_i)}{i \in [n]}.\]
  The graph \(G\) is a forest of clique-width at most \(2\),
  where the \(X_i\) have degree \(3\),
  the \(\neg X_i\) have degree \(2\),
  and all other nodes have degree \(1\).

  We set \(\Ca \deff (\vec{a}_1, \dots, \vec{a}_m) \in \bigl((V(G))^2\bigr)^m\)
  with \(\vec{a}_i = (L_{i,1}, L_{i,2})\),
  \(\sigma(i) \deff +1\) for all \(i \in [m]\),
  and
  \begin{align*}
    \phi(x_1, x_2,& y_1, \dots, y_\ell) \deff
    \Land_{i=1}^\ell \psi_{pos}(y_i) \land \Land_{i \neq j} y_i \neq y_j\\
    &\land \Lor_{i=1}^2\Biggl(\Lor_{j=1}^\ell x_i = y_j \lor \Bigl(\psi_{neg}(x_i)
  \land \Land_{j=1}^\ell \neg E(y_j, x_i)\Bigr)\Biggr)
  \end{align*}
  for \(\psi_{pos}(x) \deff \deg_{=3}(x)\) and
  \(\psi_{neg}(x) \deff \deg_{=2}(x)\) with
  \(\deg_{=k}(x) \deff \deg_{\geq k}(x) \land \neg \deg_{\geq k+1}(x)\)
  and \[\deg_{\geq k}(x) \deff \exists y_1 \cdots \exists y_k \Bigl(\Land_{i \neq j} y_i \neq y_j
  \land \Land_{i=1}^k E(x,y_i)\Bigl).\]

  For every tuple \(\vec{b} \in \set{X_1, \dots, X_n}^\ell\),
  we define an assignment \(\beta_{\vec{b}} \colon \set{X_1, \dots, X_n} \to \set{0,1}\)
  by \(\beta_{\vec{b}}(X_i) = 1\) if \(b_j = X_i\) for some \(j \in [\ell]\)
  and \(\beta_{\vec{b}}(X_i) = 0\) otherwise.
  Note that if the \(b_j\) are mutually distinct,
  then the weight of this assignment is exactly \(\ell\).
  Moreover, for all \(\vec{b} \in (V(G))^\ell\) and all \(i \in [m]\), we have
  \begin{align*}
    G \models \phi(\vec{a}_i, \vec{b}) \iff
    \ & \vec{b} \in \set{X_1, \dots, X_n}^\ell
    \text{ and the entries of } \vec{b} \text{ are}\\
    &
    \text{mutually distinct and } \beta_{\vec{b}} \text{ satisfies } L_{i,1} \lor L_{i,2}.
  \end{align*}
  Thus, \((G, \Ca, \sigma)\) is \(\phi\)-consistent if and only if
  \(\Phi\) has a satisfying assignment of weight \(\ell\).
\end{proof}
 \section{Conclusion}
\label{sec:conclusion}

Just like model checking and the associated counting and enumeration problems,
the learning problem we study here is a natural algorithmic problem for logics on finite
structures. All these problems are related, but each has its own challenges
requiring different techniques. Where model checking and enumeration
are motivated by automated verification and database systems,
we view our work as part of a descriptive complexity theory of machine learning~\cite{van_bergerem_thesis_2023}.

The first problem we studied is \uMSOLearn,
where the instances to classify consist of single vertices,
and we extended the previous fixed-parameter tractability results
for strings and trees
\cite{grohe_learning_2017-1,grienenberger_learning_2019}
to (labeled) graphs of bounded clique-width.
Moreover, on general graphs, we showed that the problem is hard for the complexity class para-NP.

For \(\MSO\)-learning problems in higher dimensions,
we presented two different approaches that yield tractability results
on graphs of bounded clique-width.
For the agnostic PAC-learning problem \PacMSOLearn,
we described a fixed-parameter tractable learning algorithm.
Furthermore, in the consistent-learning setting for higher dimensions,
we gave an algorithm that solves the learning problem
and is fixed-parameter tractable in the size of the input graph.
However, the algorithm is not fixed-parameter tractable in the size of the training sequence,
and we showed that this is optimal.

In the learning problems considered so far,
hypotheses are built using \MSO formulas and tuples of vertices as parameters.
We think that the algorithms presented in this paper for the \(1\)-dimensional case could also be extended
to hypothesis classes that allow tuples of sets as parameters.
Finally, utilizing the full power of \MSO,
one could also consider a learning problem where,
instead of classifying tuples of vertices,
we are interested in classifying sets of vertices.
That is, for a graph \(G\),
we are given labeled subsets of \(V(G)\) and want to find a hypothesis
\(h \colon 2^{V(G)} \to \set{+,-}\) that is consistent with the given examples.
It is easy to see that the techniques used in our hardness result
also apply to this modified problem, proving that it is para-NP-hard.
However, it remains open whether our tractability results
could also be lifted to this version of the problem.
 
\bibliography{main}

\end{document}